\documentclass{article}

\usepackage{dilab_arxiv}

\usepackage{enumitem}
\setlist[itemize]{itemsep=1pt, topsep=2pt, parsep=0pt, leftmargin=1.2em}
\usepackage{tabularx}
\usepackage{float}
\usepackage{placeins}
\usepackage{amsmath}
\usepackage{amsthm}
\usepackage{nicefrac}
\usepackage{url}
\usepackage{multirow}
\usepackage[capitalise]{cleveref}

\providecommand{\vu}{\boldsymbol{u}}
\providecommand{\vv}{}
\renewcommand{\vv}{\boldsymbol{v}}
\providecommand{\vp}{\boldsymbol{p}}
\providecommand{\vf}{\boldsymbol{f}}
\providecommand{\vg}{\boldsymbol{g}}
\providecommand{\vell}{\boldsymbol{\ell}}

\providecommand{\Description}[1]{}

\newcommand{\vx}{\boldsymbol{x}}
\newcommand{\vy}{\boldsymbol{y}}
\newcommand{\vq}{\boldsymbol{q}}
\newcommand{\vm}{\boldsymbol{m}}

\newtheorem{theorem}{Theorem}
\newtheorem{lemma}{Lemma}
\newtheorem{corollary}{Corollary}
\newtheorem{conjecture}{Conjecture}
\newtheorem{assumption}{Assumption}

\newtheorem{proposition}{Proposition}
\newtheorem{remark}{Remark}
\newtheorem*{theorem*}{Theorem}
\newtheorem*{lemma*}{Lemma}
\newtheorem*{proposition*}{Proposition}
\newtheorem*{corollary*}{Corollary}

\newcommand{\ALT}{\textsc{alt}}
\newcommand{\SIM}{\textsc{sim}}
\newcommand{\OSPRM}{\textsc{os-prm}$^{+}$}
\newcommand{\PRM}{\textsc{prm}$^{+}$}
\newcommand{\PP}{\textsc{pp}}
\newcommand{\TSOPT}{\textsc{ts-opt}}
\newcommand{\GSPROX}{\textsc{gs-prox}}
\newcommand{\DELAY}{\textsc{delay}}
\newcommand{\RMp}{RM$^{+}$}
\newcommand{\ALTmath}{\mathrm{ALT}}
\newcommand{\OSPRMmath}{\mathrm{OS}}

\makeatletter
\AtBeginDocument{%
  \begingroup\usefont{T1}{XCharter-TLF}{m}{n}\endgroup
  \@ifundefined{T1+XCharter-TLF}{}{%
    \DeclareFontShape{T1}{XCharter-TLF}{m}{scit}{<->ssub*XCharter-TLF/m/sc}{}}%
}
\makeatother

\title{Why Alternation Speeds Up Counterfactual Regret Minimization?}
\runningtitle{Why Alternation Speeds Up CFR}
\date{\today}

\paperlogo{\includegraphics[height=1.5cm]{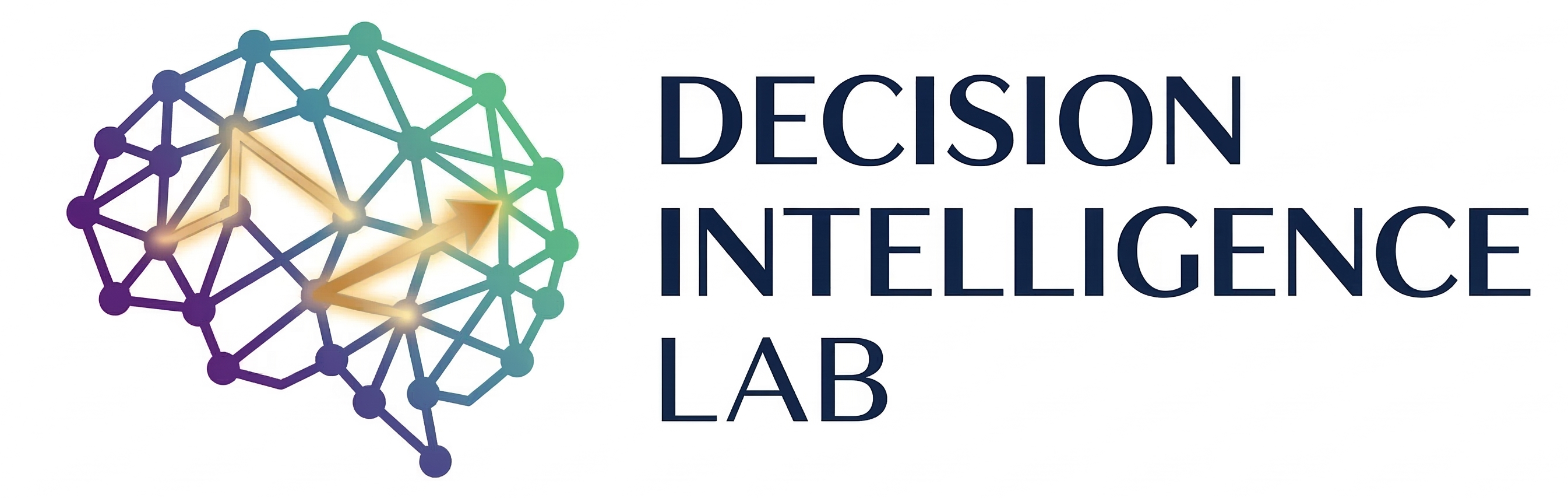}}

\author{
  Boning Li$^{1}$ and Longbo Huang$^{1\,\text{\faEnvelope}}$
  \\[0.3em]\normalfont
  $^1$Institute for Interdisciplinary Information Sciences, Tsinghua University
  \\
  \text{\faEnvelope}\ Correspondence: longbohuang@tsinghua.edu.cn
}

\begin{document}
\maketitle
\thispagestyle{fancy}

\begin{abstract}
Solving imperfect-information games at scale rests on counterfactual regret minimization
with the regret matching$^{+}$ (\RMp{}) update, and a seemingly minor implementation choice accounts for much of
its practical speed: \emph{alternating} the two players' updates, so that the second
player responds to the strategy the first player just produced. A widely repeated
intuition explains why: seeing the opponent's fresh strategy should act like optimism. We
treat this intuition as a testable hypothesis about deterministic two-player zero-sum
self-play dynamics and find that its canonical additive form fails an information-matched
test. On an analytic witness family we compare alternating \RMp{} against a one-sided
predictive \RMp{} that receives the same information channel each round and differs only
in where the fresh signal enters, a persistent regret state versus a transient prediction
term. Our main result is a dichotomy theorem that reduces trajectory separation to the
recurrence of exact-vertex play by the predictive dynamics: under an interior-band
assumption on the alternating iterate, either vertex times recur and the profile distance
is at least, infinitely often, an explicit constant that our measurements nearly attain, or they
terminate. Both antecedents are observed throughout the tested range and under multi-precision
recomputation, and a proof of either remains open. The mechanism is a readout dichotomy:
a transient injection at small persistent state turns the readout into a switch between exact vertices, while a
persistent injection is averaged into a rotation whose angular step vanishes as the
regret mass accumulates. Beyond one-sided prediction, a falsification suite over 47 matrix games and
three extensive-form witnesses rules out
the canonical implementations of three further dominant explanation routes, play-level
proximal damping, staleness interpolation, and standard two-sided optimism, whose
relative ordering crosses sign across regimes. Across the tested self-play regimes, any
correct explanation of the alternation speedup must reproduce a staleness cost present at
the smallest tested stale fraction, an advantage that persists across seven orders of payoff conditioning, regime
crossings of the optimism family, and a measured per-perturbation cost of weakening the
injection path that is far steeper for staleness than for
partial persistence.

\end{abstract}

\section{Introduction}
\label{sec:intro}

Counterfactual regret minimization (CFR) with the regret matching$^{+}$ (\RMp{}) update rule is the workhorse of
large-scale imperfect-information game solving
\citep{zinkevich2007regret,tammelin2014solving,bowling2015heads,moravvcik2017deepstack,brown2018superhuman}.
A seemingly minor implementation choice is responsible for a large share of its
practical speed. \emph{Alternating} the two players' updates lets the second player
respond to the strategy the first player just produced; under simultaneous updates both
players respond only to the previous joint profile
\citep{tammelin2014solving,tammelin2015solving,burch2019revisiting,brown2019solving}.
The speedup is robust, routinely one to two orders of magnitude, and largely unexplained.
Existing theory separates the alternating duality gap from the simultaneous one by a
nonnegative improvement term \citep{burch2019revisiting}, while the best known regret
bounds for predictive regret matching do not separate them \citep{hait2025alternating}.

There is, however, a folk explanation. Under alternation, player~2's
update sees player~1's \emph{fresh} strategy; a simultaneous update would see the one
before it. Seeing the opponent's next move is what \emph{optimistic} or \emph{predictive}
no-regret methods approximate with a prediction term
\citep{rakhlin2013online,rakhlin2013optimization,syrgkanis2015fast}, and predictive
variants of regret matching carry exactly this prediction term
\citep{farina2019stable,farina2021faster}. The reading of alternation as optimism with a
perfect one-step predictor is natural enough that, to our knowledge, it has never been
isolated and tested directly. We isolate it here, and its canonical additive form fails
the test.

We study the \emph{self-play dynamical system}: the deterministic sequence
of profiles and duality gaps generated when both players run a fixed update rule
against each other. This object is distinct from the adversarial alternating-regret
setting, whose lower bounds \citep{hait2025alternating} do not settle
self-play \emph{trajectory} equivalence. We ask whether the alternation speedup
can be \emph{reproduced} by the optimism family.

The core difficulty in testing the folk explanation is a confound: an optimistic method
differs from alternation in which information it uses and in how it uses it. We remove
the confound by construction.

On an analytic family of $2\times 2$
zero-sum games we compare two dynamics that receive the same information channel every
round, player~2's instantaneous regret against player~1's newly computed strategy,
evaluated along each dynamics' own trajectory. The two dynamics differ in one respect.
Alternating \RMp{}
(\ALT) folds this signal into the persistent regret state. One-sided predictive
\RMp{} (\OSPRM), the canonical additive realization of ``alternation as optimism,''
injects it as a transient prediction term discarded after each round. If the folk
explanation were right, the two should be close. Their duality gaps separate by a factor
of $5$ to $154$ (median $28.7$) across the family, and their trajectories remain a constant distance apart on average
(median final $\ell_1$ $0.58$).

Our main theorem turns this into mathematics. The alternating player's strategy
stays bounded away from exact vertices by a margin, verified on the tested
grid in every arithmetic precision, and under this interior-band assumption the
comparison reduces to the recurrence of the predictive dynamics' exact-vertex play.

The local analysis exposes the mechanism
(\cref{fig:phase}): a transient injection meets a persistent state that the dynamics
itself washes down to exactly zero, so the readout becomes an unbounded-gain switch
between exact vertices, while a persistent injection is averaged by a growing
normalizer into a
rotation whose angular step vanishes as the regret mass grows.

One-sided prediction is only the most literal reading of the folk explanation. We
assemble the dominant live explanation routes into a falsification catalog of four
mechanism families, labeled F1 to F4: canonical additive one-sided prediction (F1),
play-level proximal damping of the alternating structure (F2), staleness
interpolation, alternation read as a fresher simultaneous update (F3), and standard
two-sided optimism in its three standard implementations (F4). Each family comes
with its implied signature and a pre-registered falsifier, and all four canonical implementations fail on a
suite of $47$ matrix games and extensive-form witnesses (Kuhn poker, Leduc poker,
and a liar's-dice game).

Our contributions are as follows.
\begin{itemize}
\item We prove a conditional trajectory-level separation of alternation from one-sided
optimism. On an analytic witness family the canonical additive predictive form
either revisits exact vertices infinitely often, in which case its distance from
the alternating trajectory is at least an explicit constant infinitely often, a
constant our measurements nearly attain, or its vertex play terminates. Both antecedents are
observed on the tested grid (\cref{sec:theory}).
\item We identify the mechanism behind the speedup and quantify it. An exact
decomposition of the alternating duality gap isolates a nonnegative improvement
term with a closed form under plain \RMp, equal on truncation-free windows to at least a $1/\sqrt m$ fraction of the increase
of the potential that bounds the regret. On the full suite the regret channel accounts for a
median share of $0.60$ of the alternating/simultaneous (\ALT{}/\SIM{}) log-speedup, and the one-sided
predictive variant's cost falls on its opponent's regret on $42/47$ games
(\cref{sec:anatomy}, \cref{sec:exp-anatomy}).
\item We test the canonical implementations of the four dominant explanation routes
against alternation and find none equivalent. The differences are
structured: the two-sided optimism family's advantage reverses across
conditioning and tree regimes, while the price of perturbing alternation is measured: removing half the
persistence costs as little as one percent of staleness
(\cref{sec:catalog}, \cref{sec:experiments}, \cref{sec:signatures}).
\end{itemize}

\section{Related Work}
\label{sec:related}

For gradient-type dynamics alternation is provably beneficial: alternating mirror descent
improves the average regret of both players \citep{wibisono2022alternating}; alternating gradient
descent-ascent admits finite-regret and cycling analyses with near-optimal local rates
\citep{bailey2020finite,zhang2022near}, beats simultaneous updates
\citep{lee2024fundamental}, weakens the effective adversary \citep{cevher2023alternation},
and attains $O(1/T)$ average-iterate convergence in constrained bilinear games
\citep{nan2026on}. All
concern mirror-descent or gradient-descent-ascent updates; whether regret-matching-type
dynamics enjoy analogous benefits remains open. We constrain its shape.

In the adversarial model, \citet{hait2025alternating} prove $\Omega(\sqrt{T})$ lower bounds on alternating regret
for predictive regret matching$^{+}$ (\PRM{}) and optimistic gradient methods. For an alternation-aware Hedge variant,
\citet{tao2026optimalalternatingregretonline} prove logarithmic-in-$d$ bounds on
alternating regret with a matching lower bound, and average-iterate convergence at
$O(\log d/T)$ to Nash (zero-sum) and coarse correlated
(general-sum) equilibria. These guarantees concern one learner's worst-case regret under
an alternating accounting. Our object is \RMp{} self-play trajectories. In \RMp{}
self-play the alternating arm leads the information-matched predictive arm
(\cref{sec:theory}). Fictitious play separates the same two mechanisms in the opposite
direction. Without regularization, optimism alone attains constant regret on
$2\times2$ zero-sum games with a unique interior equilibrium under any tiebreaking
rule, while alternation alone incurs $\Omega(\sqrt{T})$ regret on Matching Pennies for
almost all initializations \citep{lazarsfeld2026optimism}. Our witness family lives in
exactly that game class, so the contrast isolates the update rule.

Optimistic no-regret methods exploit predictable loss sequences
\citep{rakhlin2013online,rakhlin2013optimization}, and the Regret bounded by Variation in
Utilities (RVU) property converts prediction into fast convergence
\citep{syrgkanis2015fast}. Predictive Blackwell approachability transfers this machinery
to regret matching (\PRM{}, \citealp{blackwell1956analog,farina2021faster}); optimistic
variants extend it to sequential games \citep{farina2019stable,farina2019optimistic}
and to weighted, discounted regret \citep{xu2024minimizing}, refined forms of optimism
improve rates in general games \citep{soleymani2025faster}, and adaptive asymmetry of
step sizes stabilizes predictive CFR under inaccurate predictions \citep{meng2026faster}.
In the simultaneous regime the same route has reached constant regret:
\citet{liu2026constantindividualregretgeneral} bound each player's individual regret
independently of the horizon in $N$-player general-sum self-play, via an exponential moving average (EMA) cascade
optimism of a form our two-sided falsifier does not target. These works supply the
optimistic machinery that our one-sided and two-sided families formalize, and \ALT{}'s
measured behavior contradicts both signatures.

Neither \RMp{} nor \PRM{} is last-iterate stable in general: strategies can cycle while
averages converge \citep{farina2023regret}, the last-iterate landscape is delicate
\citep{cai2025lastiterate,fiegel2025harder}, and fast last-iterate convergence provably
requires forgetting \citep{cai2024fast}. For optimistic dynamics the three notions
separate outright: a uniform best-iterate rate is available where uniform last-iterate
and polynomial random-iterate rates are not \citep{cai2026separation}. Smooth,
constrained, or perturbed variants recover guarantees
\citep{daskalakis2019last,wei2021linear,abe2023last,meng2025lastiterate}, and
optimistic regret minimization over treeplexes recovers them in sequential games
\citep{lee2021last}.
The limit points of optimistic gradient methods are characterized by
\citet{daskalakis2018limit}, reductions carry average-iterate results to the last
iterate \citep{cai2025from}, and optimistic multiplicative weights converge in the
last iterate under smoothness \citep{orabona2026lastiterate}. The separation between
plain regret matching and \RMp{} is itself sharpening: in two-player potential games
\RMp{} converges while regret matching may need exponentially many iterations
\citep{ICLR2026_b0b932e3}, and self-bounding forms tighten the regret envelope of
\RMp{} \citep{dewasurendra2026selfboundingregretmatchingpotential}. Running-average
references restore last-iterate convergence in self-play
\citep{savci2026gariprunningaveragemovingreference}; parameter-free and
scale-invariant variants of regret matching accelerate average-iterate play
\citep{zhang2025scale,meng2026a}, and a scale-invariant extragradient variant of
regret matching converges linearly in the last iterate under a norm-saturation
condition \citep{li2026last}. We contribute a sharp
instance: on the witness family, at finite horizon, the one-sided predictive dynamics exhibits a bang--bang
vertex-visiting mode with its persistent state washed to exactly zero, while alternation
rotates with a shrinking step (\cref{sec:theory}).

Alternation entered practice with CFR+
\citep{tammelin2014solving,tammelin2015solving,bowling2015heads}, was analyzed for
correctness by \citet{burch2019revisiting}, and is a default in modern discounted
variants \citep{brown2019solving,xu2024dynamic}. Recent work scales the practice up:
parallel and real-time solvers \citep{li2026gpu,li2026real}, pruning and learned or
large language model (LLM)-generated abstraction
\citep{li2024rl,li2025efficient,li2026abstraction,li2026effective}, and tooling for
Monte Carlo sampling, certification, evaluation, and tournaments
\citep{li2026av,li2026correlated,li2026agents,li2026icm}.

\section{Preliminaries}
\label{sec:setup}

A finite two-player zero-sum game is given by a payoff matrix $A \in \mathbb{R}^{m\times
n}$; player~1 plays $\vx \in \Delta_m$ (the probability simplex), player~2 plays $\vy \in \Delta_n$, and the value
to player~1 is $\vx^{\top} A \vy$. The instantaneous regret vectors are
$g^1(\vx,\vy) = A\vy - (\vx^{\top}A\vy)\mathbf{1}$ and
$g^2(\vx,\vy) = -A^{\top}\vx + (\vx^{\top}A\vy)\mathbf{1}$. Here and below,
$[\cdot]^{+}$ is the componentwise positive part and $\|\cdot\|_1$ the $\ell_1$ norm.
\RMp{} \citep{hart2000simple,tammelin2014solving} maintains a nonnegative regret state
$\vq$ updated as $\vq \leftarrow [\vq + g]^{+}$ and plays the normalized readout
$\sigma(\vq) = \vq / \|\vq\|_1$ (uniform if $\vq = \mathbf{0}$).

Each algorithm is a deterministic map $\Phi$ on a joint state (regret states plus any
auxiliary variables), iterated from zero initialization with fixed tie-breaking. A fresh
signal can enter an update in two ways: folded into the persistent regret state,
or injected as a transient prediction term used in one readout then discarded. The
protagonists are:
\begin{align}
\Phi_{\ALTmath}:\quad
& \vq^1 \leftarrow [\vq^1 + g^1(\vx_t, \vy_t)]^{+}, \quad \vx_{t+1} = \sigma(\vq^1);
\nonumber\\
& \vq^2 \leftarrow [\vq^2 + g^2(\vx_{t+1}, \vy_t)]^{+}, \quad \vy_{t+1} = \sigma(\vq^2).
\label{eq:alt}\\[2pt]
\Phi_{\OSPRMmath}:\quad
& \vq^1 \leftarrow [\vq^1 + g^1(\vx_t, \vy_t)]^{+}, \quad \vx_{t+1} = \sigma(\vq^1);
\nonumber\\
& \vq^2 \leftarrow [\vq^2 + g^2(\vx_t, \vy_t)]^{+}, \quad
\vm = g^2(\vx_{t+1}, \vy_t);
\nonumber\\
& \vy_{t+1} = \sigma([\vq^2 + \vm]^{+}).
\label{eq:os}
\end{align}
Both dynamics feed player~2 the same information channel each round, player~1's
just-computed strategy, evaluated along each dynamics' own trajectory.
In $\Phi_{\ALTmath}$ it enters the persistent state $\vq^2$; in $\Phi_{\OSPRMmath}$ the
persistent state accumulates only stale increments $g^2(\vx_t,\vy_t)$ and the fresh
signal appears as a transient prediction term $\vm$, discarded after the readout. This
is the canonical additive one-step-prediction form of \PRM{} \citep{farina2021faster}
with a \emph{perfect} predictor: the most literal realization of ``alternation as
optimism.'' The ten further arms in our catalog (simultaneous baseline \SIM, two-sided
\PRM{}, proximal damping \GSPROX$(\lambda)$, staleness \DELAY-$d$, proximal-point \PP,
two-sided optimism \TSOPT) appear in \cref{sec:catalog,app:protocol}.

We track the \emph{profile distance} $D_t = \|\vx^{\mathrm{A}}_t - \vx^{\mathrm{B}}_t
\|_1 + \|\vy^{\mathrm{A}}_t - \vy^{\mathrm{B}}_t\|_1$ between two dynamics run in
lockstep (at iteration start) and
the \emph{final gap}, the duality gap (exploitability sum) of the linearly averaged
strategies after $T$ iterations. Speed comparisons are reported as
$\mathrm{gap}_{\mathrm{arm}} / \mathrm{gap}_{\ALTmath}$ on identical games. Throughout,
$\kappa$ denotes the condition number of the payoff matrix as constructed, the ratio of its largest to
smallest singular value (a payoff shift leaves regrets unchanged but changes $\kappa$).
\emph{Conventions:} on matrix games the gap is the best-response sum
$\max_{\vx}\vx^\top A\bar\vy - \min_{\vy}\bar\vx^\top A\vy$ of the unshifted linear
averages $(\bar\vx,\bar\vy)$; on extensive-form games, where one \RMp{} state per information
set holds that set's counterfactual regrets, we report exploitability, the best-response
sum divided by two. The anatomy of
\cref{sec:anatomy} uses the shifted average $\bar\vx^{+}$ in its first identity.

We fixed \emph{ex ante} an ordered ladder of equivalence
and separation certificates: (i) \emph{dynamical equivalence} (trajectories coincide or
$D_t \to 0$); (ii) \emph{strong trajectory separation} ($D_t \ge c$ for all large $t$);
(iii) \emph{weak trajectory separation} ($D_t \ge c$ infinitely often); (iv)
\emph{performance separation} (final-gap ratio bounded away from $1$). The registration
targeted (ii) first.

\paragraph{Evidence conventions.}
The primary empirical statistic is continuous:
$\rho_g = \log_{10}(\mathrm{gap}_{\mathrm{arm}}/\mathrm{gap}_{\ALTmath})$ per game, with
medians and interquartile ranges over the suite; an arm ``reproduces'' \ALT{} only if
$\mathrm{median}\,|\rho_g| < \log_{10} 3$. Binary ``hug'' counts (trajectory closeness
plus gap-ratio band) are reported as interpretable summaries with a three-band
sensitivity analysis. The suite comprises $47$ matrix
games (12 pilot, a 20-game size family, a 15-game conditioning family) plus Kuhn, Leduc,
and a liar's-dice witness; all deterministic (fixed seed), run in lockstep against
\ALT{}.

\section{Mechanisms}
\label{sec:catalog}

We formalize the dominant live explanation routes for the alternation speedup, in the
literature and in folklore. Each route is a mechanism family with a testable signature;
we label the four families F1 to F4 in the order they appear below.
\cref{tab:catalog} summarizes the catalog; the arms' exact update rules are given in
\cref{app:protocol}.

\begin{table}[t]
\centering
\caption{The catalog of canonical explanation routes: each family with its
representative arms, the prediction we test, the tier of the test, and the question
it leaves for future explanations. The ``comparison'' tier marks a contrast between
specified algorithms.}
\label{tab:catalog}
\small
\setlength{\tabcolsep}{4pt}
\begin{tabular}{@{}p{3.0cm}p{2.85cm}p{3.7cm}p{2.3cm}p{3.4cm}@{}}
\toprule
\textbf{Family (scope)} & \textbf{Arms} & \textbf{Tested prediction} & \textbf{Test tier} & \textbf{Open question}\\
\midrule
F1: additive one-sided prediction (canonical perfect one-step form) &
\OSPRM &
$D_t \to 0$ or decaying; same-order gap &
Theorem~\ref{thm:t1} + suite &
fresh signal enters the \emph{persistent} state\\
\addlinespace[2pt]
F2: play-level proximal damping ($\lambda$-damped alternation) &
\GSPROX$(0.25)$, \GSPROX$(1.0)$ &
the cost of damping is of the same order as the cost of staleness at equal
perturbation &
empirical (suite) &
why half the persistence is as cheap as one percent of the freshness\\
\addlinespace[2pt]
F3: staleness interpolation (\ALT{} as ``fresher \SIM'') &
\DELAY-$d$, $d\in\{1,2,4,8,16\}$ ($d{=}1$ is \SIM) &
performance interpolates smoothly between \SIM{} and \ALT{} along the staleness
axis &
empirical (staleness law) &
how a damping or persistence mechanism yields the measured staleness slope\\
\addlinespace[2pt]
F4: standard two-sided optimism / RVU stabilization (three standard forms) &
\PRM{}, \PP{}, \TSOPT &
same-shape dynamics (comparable drift) and same-order speed &
empirical (three-regime cross) &
must keep the advantage across conditioning and persist on tree witnesses\\
\bottomrule
\end{tabular}
\end{table}

The first family, F1, is the most literal reading in its \emph{canonical} additive
form: a \PRM-style transient
prediction term \citep{rakhlin2013optimization,farina2021faster} with a \emph{perfect}
one-step predictor on the alternation side, making the comparison information-matched
(\cref{eq:alt,eq:os}). Making half of the fresh signal persistent already meets the median criterion, and the
cost of removing persistence grows far more slowly with the perturbation than the cost
of staleness (\cref{sec:exp-f3,sec:signatures}).

The second family, F2, starts from the resemblance of alternation to a Gauss--Seidel
sweep, since each player's update uses the other's newest value. A natural surrogate
damps the alternating
response through proximal-style mixing
$\sigma_{\mathrm{play}} = (\sigma_{\mathrm{tgt}} + \lambda\sigma_{\mathrm{old}})/(1+\lambda)$,
recovering \ALT{} exactly at $\lambda = 0$. The measured question is the price of damping at equal perturbation. The dial
reaches \ALT{} only as $\lambda\to0$, and at the smallest tested setting
($\lambda=10^{-3}$) it already costs $1.66\times$.

The third family, F3, reads alternation as ``a less stale simultaneous update.'' If
that reading were right, performance should interpolate smoothly along the staleness
axis. The \DELAY-$d$ family (player~2 sees
player~1's strategy from $d$ steps ago; $d=0$ is \ALT{}, $d=1$ is \SIM) and the mixing
dial of \cref{sec:exp-f3} span this axis.

The fourth family, F4, is standard two-sided optimism with RVU-type stability, the
strongest established acceleration
mechanism in learning in games \citep{rakhlin2013online,syrgkanis2015fast}. We test its
three standard realizations: two-sided \PRM{} \citep{farina2021faster}, a
proximal-point/extragradient-style variant \PP{} (one inner fixed-point step,
in the spirit of the ExRM$^{+}$ form of \citealp{farina2023regret}), and two-sided perfect
one-step optimism \TSOPT{} (which degenerates to \OSPRM{} restricted to one side).

Whether a family membership could produce these regime crossings is not decided by
the crossings alone; what the suite establishes is that the canonical families as
implemented do not reproduce \ALT{} at the median criterion at any catalog setting, while the weakest settings of all three dials ($s = 0.01$, $\lambda \le 0.005$, $\mu \ge 0.5$) do
(\cref{sec:experiments}), and that their advantage over \ALT{} changes sign across
conditioning and tree depth while \ALT{}'s own advantage over \SIM{} does not.

\section{Theory}
\label{sec:theory}

This section gives the conditional trajectory certificate behind F1, and the exact
local mechanics on both sides of it.

% ---------------------------------------------------------------------------
% 5b_general_mechanism.tex — 新增小节, 置于 5_theory.tex 开头（\section 行与导语之后、"Witness family" 之前）
% 证明在 A2_proofs_ext.tex。宏 \vu \vv \vp \vf \vg \vell 在 math_commands.tex。
% ---------------------------------------------------------------------------
\label{sec:general-mechanism}
Before the $2\times 2$ analysis, we isolate what is true for \emph{every} matrix game
and \emph{every} trajectory: the two ways of feeding a signal into an \RMp{} player
retain different states after the same readout. For a state $\vq \in
\mathbb{R}^{k}_{\ge 0}$ and a signal $\vg \in \mathbb{R}^{k}$ write
$U(\vq, \vg) = [\vq + \vg]^{+}$ and $\sigma(\vq) = \vq/\|\vq\|_1$ (uniform at
$\mathbf{0}$). \emph{Persistent injection} plays $\sigma(U(\vq,\vg))$ and keeps
$U(\vq,\vg)$ as the new state (this is how \ALT{} uses the fresh signal);
\emph{transient injection} plays $\sigma(U(\vq,\vm))$ and keeps $\vq$ (this is how
\OSPRM{} and every additive optimistic scheme use a prediction $\vm$). The analytic
ingredient is a normalization bound; proofs are in \cref{app:proofs-ext}.

\begin{lemma}[Normalization is $2/\|\vv\|_1$-Lipschitz at $\vv$]
\label{lem:norm-lip}
For $\vu, \vv \in \mathbb{R}^{k}_{\ge 0}$ with $\vv \neq \mathbf{0}$,
$\|\sigma(\vu) - \sigma(\vv)\|_1 \le 2\,\|\vu - \vv\|_1 / \|\vv\|_1$.
\end{lemma}

\begin{lemma}[Transient injection is scale-free at small state]
\label{lem:g1}
Let $[\vm]^{+} \neq \mathbf{0}$. Then (i) $\sigma(U(\mathbf{0}, \vm)) =
\sigma([\vm]^{+}) = \sigma([\lambda \vm]^{+})$ for every $\lambda > 0$; and (ii) for
every $\vq \ge \mathbf{0}$,
$\|\sigma(U(\vq, \vm)) - \sigma([\vm]^{+})\|_1 \le 2\,\|\vq\|_1 / \|[\vm]^{+}\|_1$.
\end{lemma}

\begin{lemma}[Persistent injection has gain at most $2/\|\vq\|_1$]
\label{lem:g2}
For every $\vq \ge \mathbf{0}$ with $\vq \ne \mathbf{0}$ and every $\vg$,
$\|\sigma(U(\vq, \vg)) - \sigma(\vq)\|_1 \le 2\,\|\vg\|_1 / \|\vq\|_1$.
\end{lemma}

\begin{theorem}[Readout interpolation; any game, any run]
\label{thm:gain}
Let $A \in \mathbb{R}^{m \times n}$ and consider player~2's \RMp{} state $\vq^2_t$ with
state readout $\vy_t = \sigma(\vq^2_t)$ in any dynamics (the played strategy under persistent injection); let $\vf_t$ be any signal with
$\|\vf_t\|_1 \le 2n\|A\|_{\max}$ (every instantaneous regret vector satisfies this).
Both injection paths produce the same immediate readout $\sigma([\vq^2_t + \vf_t]_+)$;
they differ in which state survives. The readout interpolates between the two directions
according to the mass ratio:
\begin{enumerate}
\item[(a)] If $\vq^2_t \ne 0$, then
\[
\|\sigma([\vq^2_t+\vf_t]_+) - \vy_t\|_1 \le 4n\|A\|_{\max}/\|\vq^2_t\|_1,
\]
so a move of
size $\varepsilon$ requires $\|\vq^2_t\|_1 \le 4n\|A\|_{\max}/\varepsilon$.
\item[(b)] If $[\vf_t]_+ \ne 0$, then
\[
\|\sigma([\vq^2_t+\vf_t]_+) - \sigma([\vf_t]_+)\|_1 \le 2\|\vq^2_t\|_1/\|[\vf_t]_+\|_1,
\]
and
at $\vq^2_t = 0$ the readout is the scale-free direction $\sigma([\vf_t]_+)$ however small
$\|\vf_t\|$ is. If $n = 2$ and the two components of $\vf_t$ have opposite signs,
$\sigma([\vf_t]_+)$ is a vertex.
\end{enumerate}
Which regime a dynamics occupies is determined along its trajectory:
persistent injection retains $[\vq^2_t+\vf_t]_+$ and can accumulate mass, transient
injection retains $\vq^2_t$ and accumulates only the stale increments. The dichotomy below measures the regimes on the
witness family (\ALT{} regret mass $37$--$60$; \OSPRM{} persistent state $\le 8.5\times10^{-3}$).
\end{theorem}

The theorem has no assumption on the game, horizon, or initial state: a persistently
injected player cannot jump once its regret mass is large, while a transiently injected
player jumps to the sign pattern of the current signal whenever its persistent mass is
small \emph{relative to the signal mass}. Both paths receive the same information
channel, and what differs is whether that signal can ever be read at unbounded gain. The
local readout inequalities behind both hypotheses follow from the parts: part~(b)
at zero state with $n=2$ produces exact vertices, \cref{lem:readout} gives the domination
threshold at positive state, and part~(a) the vanishing angular step of \cref{lem:complex} when
$S^2_t \to \infty$ (\cref{lem:growth}).

The analysis runs on a one-parameter witness family.

For $a \in (1, 4]$ consider the zero-sum game
\begin{equation}
A(a) = \begin{pmatrix} a & -1 \\ -1 & 1 \end{pmatrix},
\qquad x^{*} = y^{*} = \frac{2}{a+3},
\end{equation}
whose unique Nash equilibrium is interior. Parameterizing strategies by their first
component ($\vx = (x, 1-x)$, $\vy = (y, 1-y)$) and writing
$\delta^1(y) = (a+3)y - 2$ and $\delta^2(x) = 2 - (a+3)x$ for the utility differences:

\begin{lemma}[Sign structure]
\label{lem:sign}
For all $\vx, \vy$:
\[
\begin{aligned}
g^1(\vx,\vy) &= \big((1-x)\,\delta^1(y),\, -x\,\delta^1(y)\big),\\
g^2(\vx,\vy) &= \big((1-y)\,\delta^2(x),\, -y\,\delta^2(x)\big).
\end{aligned}
\]
In particular, for $y \in (0,1)$ and $\delta^2(x) \neq 0$ the two components of $g^2$
have strictly opposite signs; at $y = 1$, $g^2 = (0, -\delta^2(x))$; at $y = 0$,
$g^2 = (\delta^2(x), 0)$.
\end{lemma}

Both $\Phi_{\ALTmath}$ and $\Phi_{\OSPRMmath}$ are run from zero initialization, and the
update rules are verified against the implementation.

The next lemmas isolate how each injection path turns the same signal into
qualitatively different motion. Proofs are in \cref{app:proofs}.

\begin{lemma}[Transient readout: exact-vertex characterization]
\label{lem:readout}
In $\Phi_{\OSPRMmath}$, write $\Delta = \delta^2(x_{t+1})$ and $\vq^{2+} = (\alpha, \beta)
\ge \mathbf{0}$ for the post-update persistent state, so that $\vm = ((1-y_t)\Delta,
-y_t\Delta)$. Then
\begin{align*}
\vy_{t+1} = e_1 &\iff \beta \le y_t\Delta \ \text{ and } \ \alpha + (1-y_t)\Delta > 0,\\
\vy_{t+1} = e_2 &\iff \alpha \le -(1-y_t)\Delta \ \text{ and } \ \beta - y_t\Delta > 0,
\end{align*}
The readout vector $[\vq^{2+} + \vm]^{+}$ is $\mathbf{0}$ iff both left-hand
inequalities hold, and then $\vy_{t+1}$ is uniform by convention; if it is nonzero
and neither vertex condition holds, $\vy_{t+1}$ is interior, uniform exactly when its
two coordinates coincide. In particular an exactly
zero disfavored coordinate ($\beta = 0$ when $\Delta > 0$, $\alpha = 0$ when $\Delta < 0$)
yields a vertex for every signal of that sign with $y_t$ off the opposing boundary.
\end{lemma}

\begin{lemma}[Best-response-matched suppression]
\label{lem:washout}
In $\Phi_{\OSPRMmath}$, if $\vy_t = e_1$ and $\delta^2(x_t) \ge 0$, then the persistent
state is washed down monotonically:
$\vq^{2+}_{t+1} = \big(q^2_L,\, [q^2_R - \delta^2(x_t)]^{+}\big) \le \vq^2_t$
componentwise (and symmetrically for $e_2$ with $\delta^2 \le 0$).
\end{lemma}

\begin{lemma}[Vertex--uniform classification]
\label{lem:classify}
In $\Phi_{\OSPRMmath}$, whenever the persistent state hits zero, the next play is an
exact vertex or exactly uniform, classified completely by
$(\mathrm{sign}\,\Delta,\, y_t)$.
\end{lemma}

\begin{proposition}[Epoch structure of \OSPRM{} on the witness]
\label{prop:epoch}
Call a maximal run of consecutive rounds on which $\operatorname{sign}\delta^{2}(x_t)$ is
constant and nonzero an \emph{epoch}; it is \emph{positive} if $\delta^{2} > 0$ (favoring
$e_1$). Let $E = \{s, \dots, e\}$ be a positive epoch, let $r = q^{2}_{R,s}$ be the
disfavored persistent coordinate entering round $s$, and note that $\delta^{2}(x_{t+1}) >
0$ for $t \in \{s, \dots, e-1\}$. Then:
\begin{enumerate}
\item[(i)] \textbf{Monotone washout.} For $t \in \{s,\dots,e\}$,
$q^{2}_{R,t+1} = [\,q^{2}_{R,t} - y_t\,\delta^{2}(x_t)\,]^{+} \le q^{2}_{R,t}$ and
$q^{2}_{L,t+1} = q^{2}_{L,t} + (1-y_t)\delta^{2}(x_t) \ge q^{2}_{L,t}$. The disfavored
coordinate is exactly zero from the first $t \in E$ with $\sum_{u=s}^{t} y_u\,\delta^{2}(x_u)
\ge r$ onward.
\item[(ii)] \textbf{Absorption.} If $q^{2}_{R,t+1} = 0$ and $q^{2}_{L,t+1} > 0$ for some
$t \in \{s, \dots, e-1\}$, then $\vy_u = e_1$ and $q^{2}_{R,u} = 0$ for every $u \in
\{t+1, \dots, e\}$: every remaining round of the epoch is a vertex time.
\item[(iii)] \textbf{Absorbed epochs are finite.} If the absorption condition of (ii) is
met at some $t \in \{s, \dots, e-1\}$, then $E$ is finite and the next epoch is
negative: on an infinite run of $\vy_u = e_1$ player~1's state has $q^1_L$ strictly
increasing and $q^1_R$ nonincreasing, so $x_u$ is strictly increasing while $x_u < 1$, and
either $x_u = 1$ or $x_u \to 1$, both forcing $\delta^2(x_u) < 0$ eventually, a
contradiction.
\item[(iv)] \textbf{Recurrence criterion.} If infinitely many epochs $E$ of length at
least two satisfy
\begin{equation}
\sum_{u \in E,\, u < e} w_u\,|\delta^{2}(x_u)| \;\ge\; r_E,
\qquad w_u = \begin{cases} y_u & E \text{ positive}\\ 1 - y_u & E \text{ negative}\end{cases}
\label{eq:washout-criterion}
\end{equation}
with the favored coordinate positive at the washout round, then vertex times are
infinite and branch~(a) of \cref{thm:t1} holds (under \cref{ass:c1}).
\end{enumerate}
The statements for negative epochs follow by exchanging the roles of $L$ and $R$.
\end{proposition}

\begin{lemma}[Untruncated two-action \RMp{} update is a complex multiplication]
\label{lem:complex}
Let $\vq = (q_L, q_R)$ with $q_L, q_R > 0$, $S = q_L + q_R$, $x = q_L/S$, and let the
increment have the two-action form $\vg = \delta\,(1-x,\,-x)$ for a scalar $\delta$
(\cref{lem:sign}: every instantaneous regret vector of a two-action player has this
form). Write $z = q_L + \mathrm{i}\,q_R$ and $\varepsilon = \delta/S$. If
$\vq + \vg \ge \mathbf{0}$ componentwise, then the updated state is
\begin{equation}
z^{+} = z\,(1 - \mathrm{i}\,\varepsilon),\quad
|z^{+}|^{2} = |z|^{2} + \delta^{2}\varphi(x),\quad
\arg z^{+} = \arg z - \arctan\varepsilon,
\label{eq:complex}
\end{equation}
with $\varphi(u) = u^{2} + (1-u)^{2} \in [\nicefrac12, 1]$. The update leaves the open first quadrant
(some component of $\vq + \vg$ is nonpositive) if and only if
$\arg z - \arctan\varepsilon \notin (0, \nicefrac{\pi}{2})$. The played first
component is $x = \cos\alpha/(\cos\alpha + \sin\alpha)$ with $\alpha = \arg z$,
strictly decreasing in $\alpha$; hence $\delta > 0$ moves $x$ up and $\delta < 0$ moves
it down, by an angle of at most $|\delta|/S$.
\end{lemma}

\begin{assumption}[C0: interior play]
\label{ass:c0}
There is $T_0$ with $x^{\ALTmath}_t, y^{\ALTmath}_t \in (0,1)$ for all $t > T_0$.
\end{assumption}
Both \ALT{} states are nonzero from the first update on, so a zero
component of $\vq^i$ means vertex play, and \cref{ass:c0} says that \ALT{} truncates
neither state after $T_0$; \cref{ass:c1} implies it.

\begin{lemma}[Normalizer growth under interior play]
\label{lem:growth}
Under \cref{ass:c0}, for every $T > T_0$,
\begin{equation}
\begin{aligned}
|z^{1}_{T}|^{2} &= |z^{1}_{T_0}|^{2} + \sum_{t=T_0}^{T-1} \delta^{1}(y_t)^{2}\varphi(x_t),\\
|z^{2}_{T}|^{2} &= |z^{2}_{T_0}|^{2} + \sum_{t=T_0}^{T-1} \delta^{2}(x_{t+1})^{2}\varphi(y_t),
\end{aligned}
\label{eq:growth}
\end{equation}
and $S^{i}_{T} \ge |z^{i}_{T}|$. Consequently, for player~1 exactly one of the following
holds: (i) $\sum_t \delta^{1}(y_t)^{2} = \infty$, and then $S^{1}_t \to \infty$ and the
angular step $\arctan(|\delta^{1}(y_t)|/S^{1}_t) \to 0$; (ii) $\sum_t \delta^{1}(y_t)^{2}
< \infty$, and then $y_t \to y^{*}$. The same holds for player~2 with
$\delta^{2}(x_{t+1})$ and $x_t \to x^{*}$. In particular, under \cref{ass:c0} at least
one of the following holds: \ALT{} converges to the equilibrium, in which case
\cref{ass:c1} holds with every $\eta < \min(x^{*}, 1-x^{*})$, or at least one normalizer
diverges.
\end{lemma}

Together, \cref{lem:readout,lem:washout,lem:classify,lem:complex} give the mechanism
picture of \cref{fig:phase}. In $\Phi_{\OSPRMmath}$, playing a vertex aligned with the
stale signal washes the disfavored persistent component down (\cref{lem:washout}). A vertex is forced exactly when the disfavored persistent coordinate does not exceed
the disfavored signal mass ($y_t\Delta$ or $(1-y_t)|\Delta|$), a condition an exactly
zero coordinate satisfies for every signal of that sign (\cref{lem:readout}); so the
transient path is an unbounded-gain bang--bang readout
that switches between exact vertices, with the
oscillation concentrated on the injected side (measured $|\delta^2(x_t)|$ tail p99 $\le
0.026$). In $\Phi_{\ALTmath}$ the same channel feeds a growing normalizer. Under \cref{ass:c0},
either the normalizer diverges, with the angular step vanishing
(\cref{lem:complex,lem:growth}), or $y_t$ converges to the equilibrium. Measured, the
regret mass reaches $37$--$60$ on the tail, the readout gain $2/S$ falls accordingly, and
the measured tail-maximum first-coordinate deviation from equilibrium is $\le 0.04$ at
$T = 10^6$ (versus $0.52$--$0.71$ for \OSPRM; these are half the full profile $\ell_1$
distances). The failure of F1 lies in the injection path.

\begin{figure}[t]
\centering
\includegraphics[width=\linewidth]{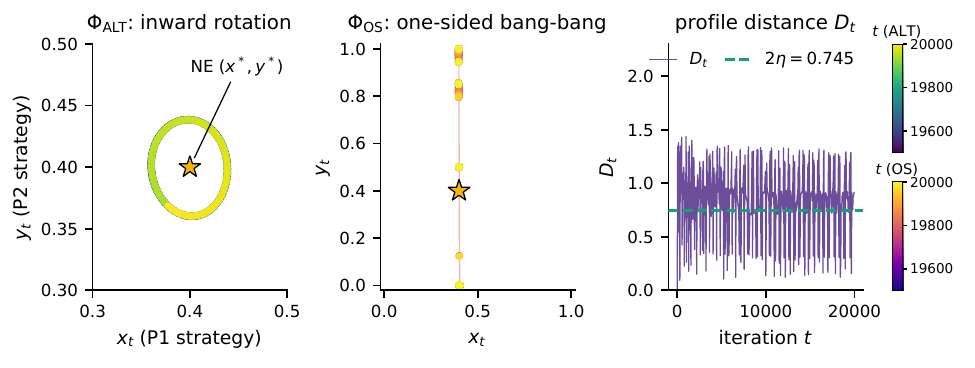}
\caption{Witness $a=2$, tail window $t \in [19501, 20000]$.
Left: \ALT{} executes an
inward rotation (measured) around $(x^{*}, y^{*}) = (0.4, 0.4)$. Middle: \OSPRM{}, updated on
the same information channel, sticks near $x^{*}$ on the un-injected side while $y$ jumps
between exact vertices. Right: $D_t$ remains of order one on the window, at least $2\eta$
at every vertex time.}
\label{fig:phase}
\Description{Three-panel figure for the two-by-two witness at $a=2$. Left: the ALT dynamics rotate inward around the interior equilibrium at (0.4, 0.4). Middle: the OS-PRM+ dynamics, updated on the same information channel each round, keep the un-injected coordinate near its equilibrium value while the injected coordinate jumps between two exact vertices. Right: the profile distance remains of order one on the window, at least the constant at every vertex time.}
\end{figure}

\Cref{prop:epoch} turns the recurrence of vertex times into a per-epoch inequality
\eqref{eq:washout-criterion} between the washing signal mass and the disfavored persistent
coordinate. On epochs of length one the sum in \eqref{eq:washout-criterion} is empty and
the readout uses the next epoch's sign; the vertex condition there is the sign-flip
inequality of \cref{lem:flip}, which the audit checks directly. On the witness grid
$54$--$93\%$ of tail epochs have length one; on the remaining epochs the inequality
holds on $43$--$100\%$ (range over $a$), length-one
epoch exits account for $26$--$86\%$ of vertex times, and the
disfavored coordinate is exactly zero at $60$--$93\%$ of vertex times (\cref{tab:epoch});
branch~(b) of \cref{thm:t1} is the event that from some round on no round meets either
vertex condition of \cref{lem:readout}: at every round with $\Delta \ne 0$ and $y_t$ off
the opposing boundary, the disfavored persistent coordinate exceeds the disfavored
signal mass.

\label{sec:dichotomy}
Call $t$ a \emph{vertex time} of $\Phi_{\OSPRMmath}$ if $\vy^{\mathrm{OS}}_{t+1}$ is an
exact vertex, i.e.\ exactly one coordinate of $[\vq^{2+}_{t+1} + \vm_t]_+$ is nonzero.
\Cref{lem:readout} gives the sufficient condition (domination: $\delta^2(x^{\mathrm{OS}}_{t+1})>0$,
$y^{\mathrm{OS}}_t<1$, $q^{2+}_{t+1,R} \le y^{\mathrm{OS}}_t\,\delta^2(x^{\mathrm{OS}}_{t+1})$, or
the symmetric one); a vertex is also retained when $\vy^{\mathrm{OS}}_t$ is already a
vertex aligned with the signal and its favored persistent coordinate is positive (\cref{lem:washout}). On the witness grid the two cases
occur at tail frequencies $0.30$ and $0.20$ (median; \cref{tab:prec}).

\begin{assumption}[C1: interior band, finite-horizon verified]
\label{ass:c1}
There exist $T_0$ and $\eta \in (0, \nicefrac12)$ such that $y^{\ALTmath}_t \in
[\eta, 1-\eta]$ and $x^{\ALTmath}_t \in (0,1)$ for all $t \ge T_0$.
\emph{Evidence:} on all 60 grid points $a \in [1.05, 4]$, over the tail half
of $T = 2\times 10^5$ iterations, $y^{\ALTmath}_t$ stays in a band with $\eta(a) \in
[0.265, 0.453]$ and $x^{\ALTmath}_t$ within $0.04$ of $x^{*} \in [0.29, 0.49]$
(\cref{tab:lastiter}), and all \ALT{} truncations are inactive, so both coordinates are
interior on the tail. The band is identical to
$2\times10^{-12}$ relative across four arithmetic precisions and persists at
$T = 10^6$ on five points. Whether it holds for all $t$ is open (\cref{conj:cinf}).
\end{assumption}

\begin{theorem}[Information-matched separation dichotomy]
\label{thm:t1}
Fix $a \in (1,4]$ and run $\Phi_{\ALTmath}$ and $\Phi_{\OSPRMmath}$ from zero
initialization. Under \cref{ass:c1}, exactly one of the following holds:
\begin{enumerate}
\item[(a)] \textbf{Weak trajectory separation.} There are infinitely many vertex times,
and at each such $t \ge T_0$,
\begin{equation}
D_{t+1} \;\ge\; \|\vy^{\ALTmath}_{t+1} - \vy^{\mathrm{OS}}_{t+1}\|_1 \;\ge\; 2\eta .
\end{equation}
\item[(b)] \textbf{Vertex visits terminate.} Only finitely many vertex times exist.
\end{enumerate}
\end{theorem}
The theorem reduces trajectory separation to a recurrence property of \OSPRM{}'s vertex
visits. The proof is in \cref{app:proofs}.

On every grid point the exact-vertex frequency of $\vy^{\mathrm{OS}}$ over the tail half
of $T=2\times10^5$ iterations is $33$--$95\%$ (median $52\%$), stable to $0.8\%$ across
four precisions and persistent at $T=10^6$; the
persistent state satisfies $\bar\varepsilon = \max_t\|\vq^{2,\mathrm{OS}}\|_\infty \le
8.5\times10^{-3}$ with the non-signal component exactly zero on $90\%$ of tail rounds
(median), the washout of \cref{lem:washout}, all at finite horizon. If
vertex times are infinite, \cref{cor:lastiter} makes the \OSPRM{} last iterate
$\Omega(1)$-exploitable infinitely often.

Two consequences hold without \cref{ass:c1}. The first is the corollary below.

\begin{corollary}[Last-iterate exploitability of \OSPRM{} on the witness; no
\cref{ass:c1}]
\label{cor:lastiter}
Fix $a \in (1,4]$. At every vertex time $t$ of $\Phi_{\OSPRMmath}$
(\cref{sec:dichotomy}), the exploitability of the last iterate
$\vy^{\mathrm{OS}}_{t+1}$ against a best response is
\[
\max_{\vx} \vx^{\top} A(a)\, \vy^{\mathrm{OS}}_{t+1} - v(a) \;\in\;
\Big\{ \tfrac{4}{a+3},\ \tfrac{(a+1)^2}{a+3} \Big\} \;\ge\; \tfrac{4}{a+3} \;\ge\;
\tfrac{4}{7},
\]
with $v(a) = \tfrac{a-1}{a+3}$.
Hence whenever vertex times are infinite (branch~(a) of \cref{thm:t1}), the
\OSPRM{} last iterate is $\Omega(1)$-exploitable infinitely often; the statement uses
only the vertex-time definition of \cref{sec:dichotomy}, not the interior-band
assumption.
\end{corollary}

The proof is in \cref{app:proofs-ext}.

The second consequence measures the gap between predicting the opponent and predicting
one's own regrets.

\begin{remark}[A perfect opponent prediction is not a perfect regret prediction]
\label{rem:shift}
\OSPRM{} predicts player~1's next strategy exactly, yet the quantity that optimistic
regret bounds control is the error between the predicted and the realized
\emph{regret vector}. Writing $\vell_{t+1} = g^2(\vx_{t+1}, \vy_{t+1})$ for the
increment that the persistent state will actually receive next,
\[
\vm_t - \vell_{t+1} \;=\; \big[u(\vx_{t+1}, \vy_t) - u(\vx_{t+1}, \vy_{t+1})\big]
\,\mathbf{1},
\]
a shift of size equal to the utility change caused by player~2's \emph{own} move.
Positive-part updates are not shift-invariant, so this error is real. The transient path
therefore carries a prediction error of this size at every vertex jump: the RVU-type machinery
\citep{syrgkanis2015fast,farina2021faster} offers no obstruction to \OSPRM{} being
fast, and the obstruction we observe is dynamical, not statistical.
\end{remark}
On the witness the error is $|y_t - y_{t+1}|\,|\delta^2(x_{t+1})|$,
small because $x$ sticks near $x^*$
($|\delta^2|$ tail p99 $\le 0.026$).

\begin{conjecture}[Unconditional separation]
\label{conj:cinf}
For every $a \in (1,4]$, \cref{ass:c1} holds and branch (b) fails from zero
initialization; hence weak trajectory separation holds unconditionally with
$c = 2\eta(a)$.
\end{conjecture}

The constant is nearly saturated in measurement.
At every one of the 60 grid points the measured minimum of $D_t$ over vertex times exceeds
$2\eta(a)$ by less than $1.2\times10^{-3}$ relative across all four precisions and never
falls below it (\cref{fig:scan}, right; at $a=2$: $0.74484$ vs.\ $0.74464$). The $y$-gap
carries essentially all of $D_t$ at vertex times.

\begin{figure}[t]
\centering
\includegraphics[width=\linewidth]{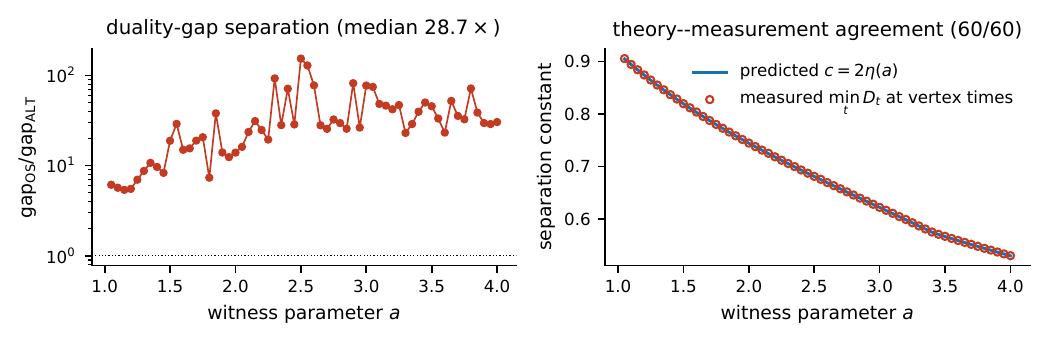}
\caption{The 60-point witness scan $a \in [1.05, 4]$. Left: duality-gap separation
$\mathrm{gap}_{\mathrm{OS}}/\mathrm{gap}_{\ALTmath}$ spans $5$--$154\times$ (median
$28.7\times$). Right: the predicted separation constant $c = 2\eta(a)$ (line) against the
measured $\min_t D_t$ over vertex times (circles); maximum relative deviation
$9.0\times 10^{-4}$ across all 60 points.}
\label{fig:scan}
\Description{Two-panel figure over the 60-point witness scan in $a$ from 1.05 to 4. Left: the ratio of the OS-PRM+ duality gap to the ALT duality gap, spanning a factor of five to a factor of 154 across the scan. Right: the predicted separation constant as a line against the measured minimum vertex-time distance as circles; the two agree up to a maximum relative deviation below one thousandth.}
\end{figure}

The registered ladder targeted strong separation first (\cref{sec:setup}); the 60-point
scan supported only weak separation, and \cref{conj:cinf} records the remaining gap
(\cref{app:ladder}). A robust variant of
\cref{lem:readout} (approximate domination yields $\rho$-near-vertex readout and constant
$2(\eta - \rho)$) is given in \cref{app:proofs}.

% ---------------------------------------------------------------------------
% 5c_anatomy.tex — 置于 5_theory.tex 末尾（dichotomy 之后）。
% 2026-10-06 R3: prop:closed / prop:order / cor:channels 的陈述按页数需要
% 从 sections/A2_proofs_ext.tex 回迁正文, 见 r3_pairs_fills.py。
% 2026-10-08 R4: prop:closed / prop:order 的证明移回 A2（见 r4_trim.py）。
% 2026-10-08 R4b: prop:closed / cor:channels 陈述、machine-check 段、
% cor:lastiter / rem:shift 的证明移回 A2（见 r4_trim2.py）。
% ---------------------------------------------------------------------------
\label{sec:anatomy}

The folk theorem bounds the duality gap of averaged play by the sum of the players'
regrets; for alternating updates the bound acquires a third term, first written down by
\citet{burch2019revisiting}. We observe that the relation is an \emph{identity} and that
the extra term has a closed form in the \RMp{} state, bounded below by a fraction of
the increase of the potential that bounds the regret. This turns ``why is alternation
fast'' into a quantity measurable along every trajectory.

Fix weights $w_t \ge 0$ with $W = \sum_{t=1}^{T} w_t > 0$ (the paper uses $w_t = t$,
the linear averaging of \cref{sec:setup}). For sequences $(\vx_t)_{t \le T+1}$,
$(\vy_t)_{t \le T}$ write $u(\vx,\vy) = \vx^{\top} A \vy$, $\bar\vx = \sum_t w_t
\vx_t / W$, $\bar\vx^{+} = \sum_t w_t \vx_{t+1}/W$, $\bar\vy = \sum_t w_t \vy_t/W$, and
\begin{align*}
R^1 &= \max_{\vx} \textstyle\sum_t w_t\,[u(\vx, \vy_t) - u(\vx_t, \vy_t)],\\
R^2_{\mathrm{stale}} &= \max_{\vy} \textstyle\sum_t w_t\,[u(\vx_t, \vy_t) - u(\vx_t, \vy)],\\
\Delta &= \textstyle\sum_t w_t\,[u(\vx_{t+1}, \vy_t) - u(\vx_t, \vy_t)],\\
R^2_{\mathrm{fresh}} &= \max_{\vy} \textstyle\sum_t w_t\,[u(\vx_{t+1}, \vy_t) - u(\vx_{t+1}, \vy)].
\end{align*}
$R^2_{\mathrm{fresh}}$ is the regret \ALT's second player actually controls (it updates
against $\vx_{t+1}$); $R^2_{\mathrm{stale}}$ is the one \SIM{} and \OSPRM{} control.

\begin{proposition}[Exact decomposition]
\label{prop:identity}
For any sequences and weights,
\begin{equation}
W \cdot \mathrm{gap}(\bar\vx^{+}, \bar\vy) = R^1 + R^2_{\mathrm{fresh}} - \Delta,
\qquad
W \cdot \mathrm{gap}(\bar\vx, \bar\vy) = R^1 + R^2_{\mathrm{stale}}.
\label{eq:identity}
\end{equation}
\end{proposition}

The uniform-weight case of the first identity is Theorem~1 of
\citet{burch2019revisiting}; both equalities are two lines of algebra
(\cref{app:proofs-ext}). Three consequences quantify the content of the relation.

The first gives the closed form of $\Delta$ under plain \RMp.

\begin{proposition}[Closed form of the improvement term under \RMp]
\label{prop:closed}
Let player~1 be plain \RMp: $\vx_t = \sigma(\vq_t)$, $\vg_t = A\vy_t - u(\vx_t,\vy_t)
\mathbf{1}$, $\vq_{t+1} = [\vq_t + \vg_t]^{+}$, $\vx_{t+1} = \sigma(\vq_{t+1})$, and let
$\mathcal{A}_t = \{i : q_{t,i} + g_{t,i} \ge 0\}$ be the untruncated coordinates. If
$S_{t+1} = \|\vq_{t+1}\|_1 > 0$ then
\begin{equation}
\begin{aligned}
(\vx_{t+1} - \vx_t)^{\top} A \vy_t
&= \frac{\langle \vq_{t+1} - \vq_t,\ \vg_t\rangle}{S_{t+1}}\\
&= \frac{\sum_{i \in \mathcal{A}_t} g_{t,i}^{2}
+ \sum_{i \notin \mathcal{A}_t} q_{t,i}\,|g_{t,i}|}{S_{t+1}}
\;\ge\; 0,
\end{aligned}
\label{eq:closed}
\end{equation}
and if $S_{t+1} = 0$ the left-hand side is $0$. Hence $\Delta \ge 0$ for every
weight sequence, recovering Theorem~5 of \citet{burch2019revisiting}, with each
summand equal to the squared untruncated regret increment plus the truncated remainder
$\sum_{i\notin\mathcal{A}_t} q_{t,i}|g_{t,i}|$, divided by the new regret mass.
\end{proposition}

The second bounds the improvement term from below.

\begin{proposition}[The improvement term is a fixed fraction of the regret
potential's increase]
\label{prop:order}
In the setting of \cref{prop:closed} with $\vq_1 = \mathbf{0}$ and uniform weights, let
$G_t^2 = \sum_{s<t}
\|\vg_s\|_2^2$. Then
\begin{equation}
\Delta_T \;\ge\; \frac{1}{\sqrt{m}} \sum_{t=1}^{T} \frac{\|\vg_t|_{\mathcal{A}_t}\|_2^2}{G_{t+1}},
\qquad
R^1_T \;\le\; \max_i q_{T+1,i} \;\le\; G_{T+1}.
\label{eq:order}
\end{equation}
On any window $[T_0, T]$ on which no truncation occurs ($\mathcal{A}_t = [m]$), the sum
telescopes: $\sum_{t=T_0}^{T} \|\vg_t\|_2^2 / G_{t+1} \ge G_{T+1} - G_{T_0}$. So on a truncation-free window $[T_0,T]$ the improvement accumulated on the window is
at least $(G_{T+1}-G_{T_0})/\sqrt{m}$, a $1/\sqrt{m}$ fraction of the \emph{increase}
over the window of the potential that upper-bounds player~1's regret; summands with
$G_{t+1}=0$ are zero.
\end{proposition}

The proof is in \cref{app:proofs-ext}.

On truncation-free windows $\Delta$ subtracts at least a $1/\sqrt{m}$ fraction of the increase of a potential
that is $O(\sqrt{T})$ for bounded payoffs, while the faster
local decay exponent of \ALT{} (\cref{sec:exp-rate}) remains open. The third
decomposes the \SIM/\ALT{} gap ratio (with the \ALT{} gap read under the
shifted-average convention) into a regret channel plus a
nonnegative improvement channel.

\begin{corollary}[Two channels of the alternation speedup]
\label{cor:channels}
For \ALT{} and \SIM{} run on the same game, horizon and weights, whenever all four
bracketed sums are positive,
\begin{equation}
\begin{aligned}
\log \frac{\mathrm{gap}_{\SIM}(\bar\vx,\bar\vy)}{\mathrm{gap}_{\ALTmath}(\bar\vx^{+},\bar\vy)}
&= \underbrace{\log \frac{(R^1 + R^2_{\mathrm{stale}})_{\SIM}}{(R^1 + R^2_{\mathrm{fresh}})_{\ALTmath}}}_{\text{regret channel}}\\
&\quad + \underbrace{\log \frac{(R^1 + R^2_{\mathrm{fresh}})_{\ALTmath}}{(R^1 + R^2_{\mathrm{fresh}} - \Delta)_{\ALTmath}}}_{\text{improvement channel}}.
\end{aligned}
\label{eq:channels}
\end{equation}
The improvement channel is nonnegative by \cref{prop:closed}; the regret channel has
no sign guarantee.
\end{corollary}

\OSPRM{} obeys the second identity in \cref{eq:identity} with no improvement term, so
its slowness must appear as inflated regrets; on $42/47$ games the inflation sits in $R^1$, where
player~1 faces a bang--bang opponent.
\cref{sec:exp-anatomy} measures all terms on the full suite.

\section{Experiments}
\label{sec:experiments}

All runs are deterministic (fixed seed, $T = 2\times 10^4$ for matrix games,
$T = 10^4$ for extensive-form games), lockstep against \ALT{}, on the 47-game matrix
suite plus three extensive-form witnesses (\cref{sec:setup}); per-game tables for the primary arms are
in \cref{app:tables}. We report each family's verdict against its pre-registered gate.

\label{sec:exp-continuous}
\cref{tab:rho} reports the primary statistic. All 11 non-\ALT{} arms have
$\mathrm{median}\,|\rho_g| \ge \log_{10} 3 \approx 0.477$; the closest is two-sided
\PRM{} at $1.158$ (a factor of $14\times$ in gap). Under the binary summary, no arm
reproduces \ALT{} on $\ge 90\%$ of games in any of three tolerance bands that
stretch the reproduction threshold two-, three-, and five-fold; the largest combined count
is \TSOPT{} at $15/47$ in the widest band. \GSPROX$(1.0)$ has zero hugs in the two
baseline bands but $6/47$ in the widest; hug
counts are heterogeneous across arms (\PP{} $13/47$, \PRM{} $11/47$, baseline band)
and concentrated in conditioning.

\begin{table}[t]
\centering
\caption{Primary continuous statistic across the 47-game matrix suite. The signed
median gives the direction, positive meaning slower than \ALT{}. No arm meets the
reproduction criterion of \cref{sec:setup}.}
\label{tab:rho}
\footnotesize
\begin{tabular}{@{}lcc@{}}
\toprule
\textbf{Arm} & $\mathrm{median}\,|\rho_g|$ & $\mathrm{median}\,\rho_g$ (signed)\\
\midrule
\SIM{} & 1.616 & $+1.616$\\
\PRM{} (two-sided) & \textbf{1.158} & $-0.886$\\
\OSPRM{} & 1.517 & $+1.517$\\
\GSPROX$(0.25)$ & 1.297 & $+1.297$\\
\GSPROX$(1.0)$ & 1.752 & $+1.752$\\
\PP{} & 2.451 & $-2.451$\\
\DELAY-1 (= \SIM) & 1.616 & $+1.616$\\
\DELAY-2 & 1.897 & $+1.897$\\
\DELAY-4 & 2.221 & $+2.221$\\
\DELAY-8 & 2.673 & $+2.673$\\
\TSOPT{} & 1.615 & $-1.615$\\
\bottomrule
\end{tabular}
\end{table}

\label{sec:exp-f1}
Across the 47-game suite, \OSPRM{} is slower than \ALT{} by a median $33\times$ in final
gap, with drift $66\times$ and oscillation $263\times$ larger; the profile
distance stays bounded away from zero on average (median final $L_1$ $0.58$): same fresh information,
different injection path.

\label{sec:exp-f3}
Both axes of F3 fail. \cref{fig:delay} shows the gap ratio to \ALT{} against the stale
fraction $s$ on both: the integer delay family ($s=d$ for $d$ in $\{1,2,4,8,16\}$,
medians $41.3$ to $1408$, per-game ratios in \cref{tab:delayfull}) and the
mixing dial ($s$ from $0.01$ to $1$, medians rising from $2.62$ to $41.3$); a log-space least-squares fit gives $\mathrm{ratio}-1\approx 36\,s^{0.67}$
(largest deviation $25\%$; the two largest-$s$ medians coincide to two decimals in $\log_{10}$, $17/47$ games
have $r(0.9)>r(1)$). Both branches trace one rising curve, the delay branch with median
per-game per-doubling ratios $1.9, 2.1, 2.3, 2.4$; Kuhn and Leduc give the same shape
($5.3\times$ and $5.6\times$
at $s=0.01$; $119\times$ and $321\times$ at $s=1$). On 50 independent random games the
$s=1$ median is $39.4$. At the smallest tested
$s=0.01$ the penalty is already $2.6\times$ ($26/47$ games within $3\times$); below
$s=0.01$ is untested.

\begin{figure}[t]
\centering
\begin{minipage}[t]{0.48\linewidth}
\centering
\includegraphics[width=\linewidth]{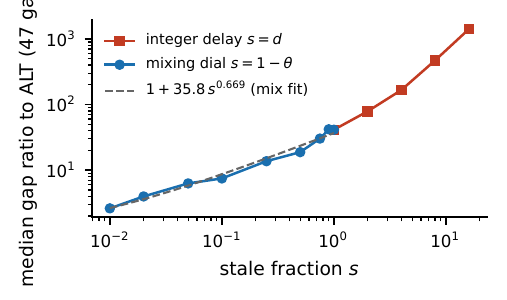}
\caption{Gap ratio to \ALT{} against the stale fraction $s$ on log--log axes; the
delay and mixing branches trace one curve.}
\label{fig:delay}
\end{minipage}\hfill
\begin{minipage}[t]{0.48\linewidth}
\centering
\includegraphics[width=\linewidth]{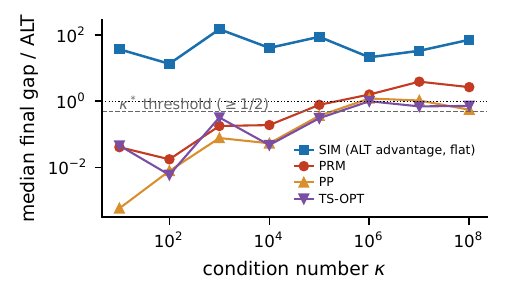}
\caption{Median gap ratio to \ALT{} against the conditioning value $\kappa$, for
\ALT{}/\SIM{} and for the optimism family; 5 instances per $\kappa$.}
\label{fig:kappa}
\end{minipage}
\Description{Two side-by-side single-column panels. Left: gap ratio against the stale fraction on log-log axes; the mixing dial and the integer-delay family lie on one rising curve down to the smallest tested stale fraction. Right: median advantage ratios against the conditioning parameter on a logarithmic horizontal axis; the ALT-versus-SIM advantage stays large across all conditioning values while the optimism family's advantage closes at strong conditioning and one arm reverses to slower than ALT.}
\end{figure}

\label{sec:exp-f2}
As $\lambda\to0$, \GSPROX$(\lambda)$ converges to \ALT{} on any finite horizon
(the update is continuous in $\lambda$), so the question is the shape of the approach. On a
grid $\lambda\in\{10^{-3},\dots,4\}$ the median gap ratio rises from $1.7\times$ at
$\lambda=10^{-3}$ to $152\times$ at $\lambda=4$, fit by
$\mathrm{ratio}-1 \approx 52\,\lambda^{0.64}$
(log-space least squares). A damping-stabilized explanation predicts a range of $\lambda$
on which damped alternation matches \ALT{}; at the smallest tested $\lambda=10^{-3}$,
$39/47$ games are within $3\times$ of \ALT{} and $1\%$ damping already costs
$4\times$ ($15/47$). Below $\lambda=10^{-3}$ is untested.

\label{sec:exp-f4}
On well-conditioned matrix games, the F4 arms are \emph{faster} than \ALT{} (medians:
\PRM{} $7.7\times$, \TSOPT{} $41\times$, \PP{} $280\times$), which alone rules out
equivalence in one direction; their drift statistics differ from \ALT's by factors of
$0.03$--$0.47$, so the dynamics are not the same shape (a Spearman speedup--drift
correlation, $n{=}12$, is sign-unstable across horizons and never significant;
auxiliary only).

The discriminating evidence is a cross of three regimes. First, conditioning (\cref{fig:kappa}): sweeping
$\kappa \in \{10^1, \dots, 10^8\}$ (5 instances each), the \ALT/\SIM{} ratio stays
$\ge 13.5\times$ in every band with no trend, while the optimism family's advantage
over \ALT{} closes at $\kappa^{*} = 10^5$--$10^6$ and \PRM{} \emph{reverses} to
$3.9\times$ and $2.7\times$ slower than \ALT{} at
$\kappa = 10^7, 10^8$. The closure is
operationalization-dependent: it does not trigger within $\kappa \le 10^8$ under a
vs-\SIM{} reading, and the hypothesis under test is equivalence to \ALT. Second,
tree witnesses (\cref{sec:exp-efg}).

\label{sec:exp-efg}
All arms cost one tree pass per player per iteration except \PP{} (two).
On Kuhn poker the matrix-level picture persists qualitatively for the slower arms
(\OSPRM{} $75\times$ slower; one-step staleness $119\times$; \TSOPT{} the only arm
within the $3\times$ band), while two-sided \PRM{} and \PP{} are
$14\times$ and $520\times$ \emph{faster} than \ALT{} here. On Leduc poker the picture sharpens into a reversal: \ALT{} attains the
lowest exploitability of all 12 arms, with the two-sided optimism arms $19$--$70\times$
slower (\cref{tab:efgfull}) and \OSPRM{} $1.5\times 10^{3}$ slower, consistent with
reports that predictive CFR variants underperform CFR$^{+}$ on poker
\citep{farina2021faster,farina2023regret,li2026pokerskill}.

The third witness, a liar's-dice game \citep{neller2011approximating} with 8177 nodes
and 1024 information sets, was pre-registered with the gate ``the Leduc-level
separation persists,'' and this gate failed: the game is easy for every
dynamics (nine of the twelve arms reach exploitability $10^{-8}$--$10^{-6}$ at
$T = 10^4$), compressing the separation band (four arms inside the $3\times$ gap band
of \ALT{}). Two structures survive at this size: \ALT{} is the fastest of all 12 arms,
and no arm hugs \ALT{} (final $L_1$ distances $10.1$--$23.3$, flat). The nominal ordering is not scale-stable: on two larger trees ($51{,}176$ and
$294{,}877$ nodes) \OSPRM{} and \textsc{alt-prm}$^{+}$ go from edging \ALT{} out by up
to $2\%$ to trailing it by $9$--$13\%$, while the trajectory separation holds at both
scales in raw totals. The tree-persistence claim is
therefore witness-level (Kuhn/Leduc), with liar's dice a boundary game; the staleness law is unaffected
($8\times$ at $d{=}1$ rising to $4.4\times 10^3$ at $d{=}8$).

\label{sec:exp-big}
On the $8177$-node witness the arms inside the gap band sit at \PP{} $1.07\times$,
\TSOPT{} $1.37\times$, \GSPROX(0.25) $1.57\times$ and \OSPRM{} $1.88\times$ \ALT's
gap, and \DELAY-2/4/8 stop at $2.1\times 10^{-6}$, $2.4\times 10^{-5}$ and
$2.4\times 10^{-4}$ exploitability. The two larger trees were pre-registered with no
PASS/FAIL; the ordering itself is the reading. Per-information-set distances fall
(\textsc{alt-prm}$^{+}$ from $0.00276$ to $0.00170$), every non-\ALT{} arm has
tail-mean $L_1 \ge 14$ ($\ge 41$ on the larger tree) from \ALT, and the decomposition
of \cref{prop:identity} gives \OSPRM's $R^1$ at $1.07\times$ \ALT's on the smaller
tree and $1.34\times$ on the larger, while its stale-sequence regret is driven to
zero.

\begin{table}[htbp]
\centering
\caption{Scale test on the $51{,}176$-node liar's-dice tree ($5{,}120$ information
sets, $T = 10^4$): gap of
the average strategy, ratio to \ALT, tail-mean $L_1$ distance to \ALT's profile, and
the exact decomposition of \cref{eq:identity} per unit weight.}
\label{tab:big5}
\small
\begin{tabular}{@{}lccccc@{}}
\toprule
& \SIM & \ALT & \OSPRM & \PRM & \textsc{alt-prm}$^{+}$\\
\midrule
gap$(\bar\vx,\bar\vy)/2$ ($10^{-7}$) & 3.420 & 0.675 & 0.661 & 1.327 & 0.669\\
gap$(\bar\vx^{+},\bar\vy)/2$ ($10^{-7}$) & 3.324 & 0.577 & 0.565 & 1.230 & 0.564\\
ratio to \ALT & 5.07 & 1 & 0.98 & 1.97 & 0.99\\
tail $L_1$ to \ALT & 167.2 & 0 & 121.9 & 168.5 & 14.1\\
tail $L_1$ per info set & 0.0327 & 0 & 0.0238 & 0.0329 & 0.00276\\
$R^1/W$ ($10^{-7}$) & 3.781 & 1.163 & 1.246 & 1.336 & 1.247\\
$R^2_{\mathrm{fresh}}/W$ ($10^{-7}$) & 6.710 & 0.917 & 1.100 & 3.437 & 1.023\\
$R^2_{\mathrm{stale}}/W$ ($10^{-7}$) & 3.059 & 0.188 & 0.077 & 1.318 & 0.090\\
$\Delta/W$ ($10^{-7}$) & 3.843 & 0.926 & 1.216 & 2.312 & 1.142\\
\bottomrule
\end{tabular}
\end{table}

\begin{table}[htbp]
\centering
\caption{Same quantities on the $294{,}877$-node liar's-dice tree ($24{,}576$
information sets, $T = 10^4$).}
\label{tab:big6}
\small
\begin{tabular}{@{}lccccc@{}}
\toprule
& \SIM & \ALT & \OSPRM & \PRM & \textsc{alt-prm}$^{+}$\\
\midrule
gap$(\bar\vx,\bar\vy)/2$ ($10^{-7}$) & 2.476 & 0.789 & 0.863 & 1.261 & 0.895\\
gap$(\bar\vx^{+},\bar\vy)/2$ ($10^{-7}$) & 2.365 & 0.669 & 0.750 & 1.150 & 0.770\\
ratio to \ALT & 3.14 & 1 & 1.09 & 1.60 & 1.13\\
tail $L_1$ to \ALT & 1551.7 & 0 & 1403.6 & 1546.9 & 41.9\\
tail $L_1$ per info set & 0.0631 & 0 & 0.0571 & 0.0629 & 0.00170\\
$R^1/W$ ($10^{-7}$) & 2.695 & 1.311 & 1.753 & 1.331 & 1.857\\
$R^2_{\mathrm{fresh}}/W$ ($10^{-7}$) & 4.807 & 1.079 & 1.437 & 3.039 & 1.501\\
$R^2_{\mathrm{stale}}/W$ ($10^{-7}$) & 2.257 & 0.267 & $-0.028$ & 1.191 & $-0.068$\\
$\Delta/W$ ($10^{-7}$) & 2.774 & 1.052 & 1.691 & 2.070 & 1.818\\
\bottomrule
\end{tabular}
\end{table}

% ---------------------------------------------------------------------------
% 6b_experiments_ext.tex — 扩展实验小节（接在 6_experiments.tex 的 6.6 之后, 同一 \section 内）。
% 2026-10-03 页数裁剪: 原 6.8（injection dial）与 6.9/6.10/6.11 曾移入补充册 sections/F_ext_results.tex。
% 2026-10-14 arXiv 版: 已全部搬回正文 sections/6_experiments.tex（F_ext_results.tex 已删除）。
% 数字全部取自 ext/results/SUMMARY_EXT.md; 判读句照录 PREREG_EXT.md。
% ---------------------------------------------------------------------------
\label{sec:exp-anatomy}
\cref{prop:identity} lets us read every final gap as a sum of measured terms. Both
channels contribute (pre-registered reading). On the 47
games, the \ALT{}/\SIM{} log-speedup has median $1.60$ in $\log_{10}$, of which the
regret channel contributes $0.91$ (IQR $0.67$--$1.17$) and the improvement
channel $0.58$ (IQR $0.42$--$0.73$); the regret channel's share is $0.60$ (median),
above one half on $31/47$ games.

For \OSPRM{}, the second identity has no improvement term, and
the inflation sits mainly in $R^1$: player~1's regret is $23.0\times$
\ALT's (median), while player~2's stale regret is $-8.5\times$ \ALT's fresh
regret (median), negative on $40/47$ games, where player~1 faces a bang--bang opponent.
The transient prediction minimizes
the stale-sequence regret well, at the opponent's expense.
\cref{tab:anatomy} tabulates both gap conventions on the witness games and
\cref{fig:anatomy} breaks the suite down game by game.

\begin{figure}[t]
\centering
\includegraphics[width=\linewidth]{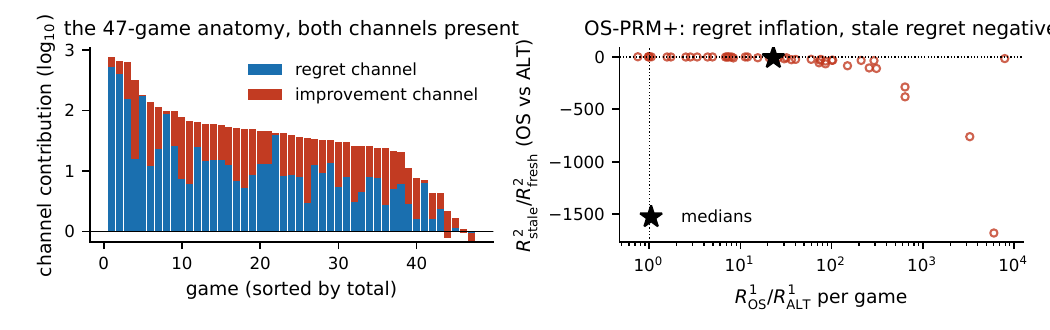}
\caption{Per-game anatomy of the \ALT/\SIM{} log-speedup on the 47-game suite
($T = 2\times 10^{4}$). Left: regret-channel and improvement-channel contributions,
stacked and sorted by total; both channels contribute on $44/47$ games (median
share of the regret channel $0.60$; the improvement channel is nonnegative on every
game, as \cref{prop:closed} requires). Right: the \OSPRM{} regret anatomy, one point
per game; the player-1 regret ratio exceeds one on $42/47$ games (median
$23.0\times$) while the stale-regret ratio is negative on $40/47$ (below the dotted
line); medians are marked.}
\label{fig:anatomy}
\Description{Two panels. Left: stacked bar chart of the two channel contributions for each of 47 games, sorted by total; most bars contain both a regret-channel and an improvement-channel segment. Right: scatter of per-game player-1 regret ratio on a logarithmic horizontal axis against the stale-regret ratio on a linear vertical axis; most points lie right of one and below zero, with medians marked by a star.}
\end{figure}

% 2026-10-14 arXiv 版: 6.8/6.9/6.10/6.11（injection dial、可加性、视界、敏感性/末迭代/初始化）
% 已搬回正文 sections/6_experiments.tex，紧接本文件之后。

\label{sec:exp-inj}
\textsc{inj}$(\mu)$ keeps the fresh signal's total readout weight fixed at one and moves
a fraction $\mu$ of it into the persistent state (\cref{app:protocol}); $\mu = 0$ is
\OSPRM{} and $\mu = 1$ is \ALT{}, bit-exactly. Median $\rho_g$ against \ALT{} is
$+1.52$ at $\mu=0$, $+1.61$ at $\mu=0.1$, $+1.51$ at $\mu=0.25$, $+0.22$ at $\mu=0.5$,
$-0.03$ at $\mu=0.75$, $-0.02$ at $\mu=0.9$, and exactly $0$ at $\mu=1$; over the same
range the trajectory distance to \ALT{} (tail mean) falls from $0.69$ to $0.07$ and the
exact-vertex frequency from $0.12$ to $0$. The pre-registered reading is ``persistence
at $\mu=0.5$ already reproduces \ALT{} (median within $3\times$)'' (within $3\times$ on
$27/47$ games, against $8/47$ at $\mu=0$): making half of the fresh signal persistent
suffices. The full per-$\mu$ grid of both quantities is \cref{tab:inj} and
\cref{fig:inj}.
On Kuhn and Leduc the dial is more forgiving: the median falls from $+2.53$ at $\mu=0$
to $+0.66$ at $\mu=0.5$ and $\le +0.36$ for $\mu \ge 0.75$; the pre-registered reading
there is ``gradual: reproduces \ALT{} only for $\mu > 0.5$''.

\begin{table}[tbp]
\centering
\caption{The injection dial, by persistent fraction $\mu$. The grid $\{0, 0.1, 0.25, 0.5, 0.75, 0.9, 1\}$ was pre-registered; intermediate values were added after a pilot run. $\mu=0$ is \OSPRM{}, $\mu=1$ is \ALT{}.}
\label{tab:inj}
\footnotesize
\setlength{\tabcolsep}{3.5pt}
\begin{tabular}{@{}lcccc@{}}
\toprule
$\mu$ & median $\rho_g$ [IQR] & within $3\times$ & $D$-tail & vertex\\
\midrule
0.00 & $+1.52$ [$0.96$, $2.08$] & 8/47 & $0.686$ & $0.119$\\
0.10 & $+1.61$ [$0.98$, $2.08$] & 9/47 & $0.669$ & $0.137$\\
0.25 & $+1.51$ [$1.09$, $2.10$] & 9/47 & $0.686$ & $0.251$\\
0.30 & $+1.46$ [$1.06$, $2.04$] & 9/47 & $0.650$ & $0.250$\\
0.35 & $+1.32$ [$0.51$, $1.81$] & 12/47 & $0.588$ & $0.233$\\
0.40 & $+1.20$ [$0.14$, $1.67$] & 14/47 & $0.508$ & $0.099$\\
0.45 & $+0.58$ [$-0.02$, $1.05$] & 21/47 & $0.320$ & $0.000$\\
0.50 & $+0.22$ [$-0.08$, $0.57$] & 27/47 & $0.072$ & $0.000$\\
0.75 & $-0.03$ [$-0.11$, $0.08$] & 43/47 & $0.054$ & $0.000$\\
0.90 & $-0.02$ [$-0.13$, $0.02$] & 46/47 & $0.044$ & $0.000$\\
1.00 & $+0.00$ [$0.00$, $0.00$] & 47/47 & $0.000$ & $0.000$\\
\bottomrule
\end{tabular}
\end{table}

\begin{figure}[t]
\centering
\includegraphics[width=\linewidth]{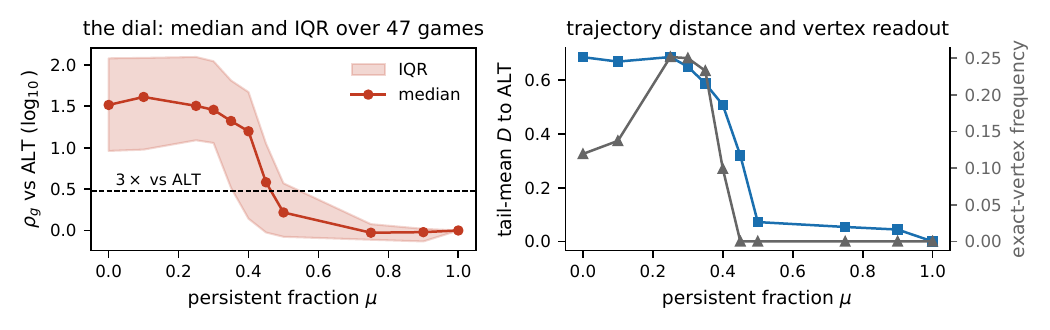}
\caption{The injection dial on the 47-game suite. Left: median $\rho_g$ against \ALT{}
with IQR band, per persistent fraction $\mu$; the dashed line is the $3\times$
reproduction threshold. Right: tail-mean trajectory distance to \ALT{} (squares, left
axis) and exact-vertex frequency (triangles, right axis); the tail distance falls from
$0.69$ to $0.07$ at $\mu = 0.5$ and the vertex frequency reaches $0$ there.}
\label{fig:inj}
\Description{Two panels against the persistent fraction mu from zero to one. Left: a median line with a shaded interquartile band, starting near one and a half, crossing the dashed reproduction threshold at mu of about 0.45, and reaching zero at mu of one. Right: two curves, a tail-mean trajectory distance and an exact-vertex frequency, both falling from their maxima at mu of zero to zero at mu of one half.}
\end{figure}

\label{sec:exp-add}
Composing the two updates measures whether their gains add. \textsc{alt-prm}$^{+}$ (alternating \PRM{}, the schedule of
PCFR$^{+}$ \citep{farina2021faster}) attains a median log-gain over \SIM{} of $3.92$
against $1.62$ for \ALT{} and $2.41$ for \PRM{}; the gain over the better single
component is $0.50$ in $\log_{10}$ (median, positive on $37/47$ games). The
pre-registered reading is ``combination beats the better component by $\ge 2\times$'':
the two updates compose.

\label{sec:exp-rate}
All suite statistics above are at $T = 2\times 10^{4}$. Running to $2\times 10^{5}$,
the median $\rho_g$ of \OSPRM{} moves from $1.52$ to $1.71$ (per-game growth median
$+0.24$ in $\log_{10}$), of \SIM{} from $1.62$ to $1.71$, and of \PRM{} from $-0.89$
to $-2.31$. The local decay exponent between the two horizons is $-0.82$ for \ALT{}
against $-0.50$ for \OSPRM{} (\ALT{} steeper on $35/47$ games, sign test
$p = 5.4\times 10^{-4}$; \cref{fig:rate}). Both pre-registered rate readings fire on
this comparison: the separation grows with $T$ (rate-level, empirical), and \ALT{}
decays with a faster local exponent.
On the witness, the local exponents of \ALT{} at $a=2$ are $(-0.58, -1.08, -0.57)$
across the three decades up to $2\times 10^{6}$, against $(-0.49, -0.51, -0.51)$ for
\OSPRM{} and $(-0.49, -0.49, -1.16)$ for \SIM{}. Whether alternating \RMp{} enjoys an
$O(1/T)$ guarantee of the kind proved for the algorithm class of
\citet{tao2026optimalalternatingregretonline} is open; our exponents are local
measurements, not rates.

\begin{figure}[t]
\centering
\includegraphics[width=\linewidth]{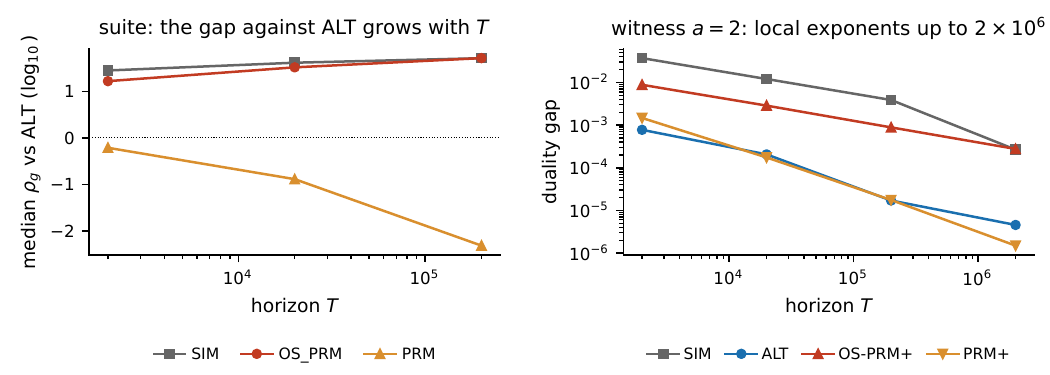}
\caption{Horizon dependence. Left: median $\rho_g$ against \ALT{} on the suite at
$T = 2\times 10^{3}, 2\times 10^{4}, 2\times 10^{5}$; the separation against \OSPRM{}
grows with $T$ while \PRM{} becomes relatively faster. Right: duality gaps on witness
$a = 2$ up to $T = 2\times 10^{6}$, the range over which the local exponents are
measured.}
\label{fig:rate}
\Description{Two panels. Left: three lines against horizon on a logarithmic horizontal axis showing the median log-ratio to ALT: the OS-PRM+ line rises, the SIM line is nearly flat, and the PRM+ line falls. Right: log-log plot of duality gap versus horizon up to two million iterations for four arms on the witness game, with all gaps decaying steadily across the six decades.}
\end{figure}

\label{sec:exp-sens}
\cref{thm:gain} bounds the readout gain of a persistently injected player and leaves a
transiently injected player's gain unbounded at zero state; the trajectory-level
sensitivity below is measured, not a consequence of the theorem. Perturbing the payoff matrix by a relative $10^{-13}$ (which
leaves exact \RMp{} trajectories unchanged by scale invariance) moves \ALT{}
trajectories by at most $2.8\times 10^{-11}$ in $L_1$ over $2\times 10^{4}$ rounds, and
\OSPRM{} trajectories by a median of $2.00$, with macroscopic divergence ($L_1 > 0.1$)
on $33/47$ of the games, first reached after a median of $182$ rounds; \SIM{} and \PRM{}
give $1/47$ and $5/47$ (\cref{tab:sens}). The same perturbation applied to the initial
regrets gives $31/47$ for \OSPRM{} (first at a median of $97$ rounds) and $0/47$ for
\ALT{} (largest $L_1$ deviation $1.0\times 10^{-6}$).
On the witness, the \OSPRM{} last iterate is an exact vertex on $52$--$82\%$ of tail
rounds at $T = 10^{6}$ (per $a$), hence at least $4/(a+3) \ge 4/7 \approx 0.57$
exploitable against a best response (\cref{cor:lastiter}); its tail-maximum
first-coordinate deviation from equilibrium is $0.52$--$0.71$, i.e.\ a full profile
$\ell_1$ distance of $1.05$--$1.43$ (\cref{fig:lastiter,tab:lastiter}), while \ALT's
last iterate is
within $0.04$ and never a vertex: last-iterate separation on all four witness points.
Finally, from $20$ random initializations $\vq \sim U[0,1]$, the F1 verdict (median
$|\rho_g| \ge \log_{10} 3$) holds on all $20$ seeds (minimum per-seed median $0.76$),
and the interior band of \cref{ass:c1} has $\eta \ge 0.357$ with \OSPRM{} vertex
frequency $\ge 0.65$ on every seed.

\begin{table}[tbp]
\centering
\caption{Tail-maximum last-iterate deviation on the witness family, $T=10^{6}$
(\cref{app:lastiter-metric}): profile $\ell_1$ is twice the coordinate deviation.}
\label{tab:lastiter}
\footnotesize
\begin{tabular}{@{}lcccc@{}}
\toprule
$a$ & $y^{*}$ & \ALT{} coord.\ dev. & \OSPRM{} coord.\ dev. & \OSPRM{} profile $\ell_1$\\
\midrule
1.2 & 0.4762 & 0.039 & 0.5245 & 1.049\\
2.0 & 0.4000 & 0.027 & 0.6005 & 1.201\\
3.0 & 0.3333 & 0.022 & 0.6671 & 1.334\\
4.0 & 0.2857 & 0.020 & 0.7148 & 1.430\\
\bottomrule
\end{tabular}
\end{table}

\begin{figure}[t]
\centering
\includegraphics[width=0.48\linewidth]{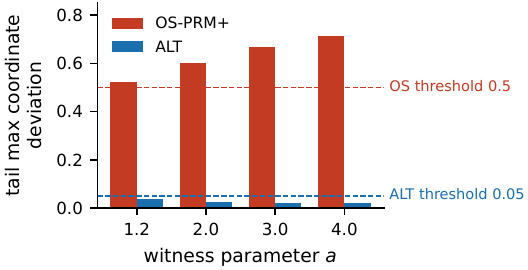}
\caption{Last-iterate separation on the witness family: maximum first-coordinate
deviation from equilibrium; the full profile $\ell_1$ distance is twice this value.
Both are measured over the tail half of $10^{6}$ iterations, per witness point. Dashed
lines mark the pre-registered thresholds (\OSPRM{} $\ge 0.5$, \ALT{} $\le 0.05$), met
at all four $a$.}
\label{fig:lastiter}
\Description{Grouped bar chart at four witness parameter values. At each value the OS-PRM+ bar is above one half while the ALT bar is below one twentieth; dashed horizontal lines mark those two thresholds.}
\end{figure}

\section{Findings}
\label{sec:signatures}

The catalog's failures leave measured regularities, scoped to the tested self-play regimes
(47-game suite; Kuhn, Leduc, liar's dice).

Staleness is expensive down to the smallest tested fraction. Let $s$ be the
stale fraction of the information player~2 uses, so $s = 1-\theta$ for the mixing dial
$\textsc{mix}(\theta)$ and $s=d$ for the integer delay family.
On the mixing branch the
median gap ratio to \ALT{} is fit by $\mathrm{gap}/\mathrm{gap}_{\ALTmath}-1 \approx
36\,s^{0.67}$ over $s\in[0.01,1]$ (log-space least squares, largest deviation $25\%$;
\cref{fig:delay}); on the delay branch it grows slightly faster than linearly
(median per-game per-doubling ratio $1.9\to2.4$ from $d=2$ to $16$). One percent of staleness already
costs $2.6\times$ (median; within $3\times$ of \ALT{} on $26/47$ games), and $100\%$
costs $41\times$. Any
theory of alternation as ``a slightly fresher
simultaneous update'' must reproduce this cost at $s=0.01$.

The speedup survives payoff conditioning. Across seven orders of
magnitude of condition number the \ALT{} advantage over \SIM{} never falls below
$13.5\times$ and shows no monotone trend (medians $13.5$ to $149.1$ across the eight
bands), while optimism-family acceleration closes at
$\kappa^{*} \approx 10^5$--$10^6$ and partially reverses.

Optimism's advantage comes and goes with the regime, while alternation's holds
throughout. At the
suite-median level the \ALT/\SIM{} ratio is $\ge 13.5\times$ at every $\kappa$,
$119\times$ on Kuhn, $321\times$ on Leduc, and $8$/$5.1$/$3.1\times$ on the three
liar's-dice trees. The two-sided optimism arms beat \ALT{} on well-conditioned matrix
games and on Kuhn (\PRM{} $14\times$, \PP{} $520\times$ faster), lose to it on Leduc
($19$--$70\times$ slower), \PRM{} reverses to $3.9\times$ and $2.7\times$ slower at $\kappa=10^7,10^8$,
\PP{} is within $1.22\times$ of \ALT{} at $\kappa=10^6,10^7$, and \TSOPT{} never
falls behind, and \OSPRM{} straddles it on liar's dice ($0.98\times$ at
$51{,}176$ nodes, $1.09\times$ at $294{,}877$). An arm whose dynamics coincided with \ALT's could not change sign relative to it
across these regimes.

Weakening the fresh signal costs less when the injection path is degraded than
when freshness is. Under the median criterion every dial
reaches \ALT{} at its limit, and the cost at the smallest setting is
$2.6\times$ for $1\%$ staleness, $1.66\times$ for $0.1\%$ damping, and
$1.65\times$ for a $50\%$ loss of persistence. Removing half the persistence is as
cheap as removing one percent of the freshness; short of the \ALT{} endpoint, only the injection dial puts $90\%$ of games within $3\times$, from $\mu=0.75$ on
(\cref{tab:dials}). Freshness is the expensive resource, and any
explanation must reproduce this cost curve.

With the conditional certificate of \cref{thm:t1}, these four regularities
form a checklist for explanations of the speedup. Positions untouched
by the four comparisons include persistent auxiliary prediction states outside F1,
non-RVU stabilization, and per-round fresh-information analyses,
as the staleness cost supports.

\section{Conclusion}
\label{sec:conclusion}

The folk explanation of the alternation speedup, seeing the opponent's fresh
strategy as optimism, fails an information-matched test. On the $2\times2$
witness family, under an interior-band assumption on the alternating iterate, either
exact-vertex times of the predictive dynamics recur and the profile distance is at least
a constant between $0.53$ and $0.91$ infinitely often, or they terminate; the
antecedents hold on all 60 grid points at $2\times10^5$ iterations across four
precisions. The mechanism is a readout dichotomy: transient injection at small persistent state produces an
unbounded-gain bang--bang readout, persistent injection a rotation whose angular step
vanishes as the measured regret mass grows. Three further
routes (play-level damping, staleness interpolation, two-sided optimism) are not
reproduced in canonical implementations on the 47-game suite and tree witnesses,
and weakening them costs more when freshness is lost than when the injection path
is degraded. Within the tested regimes, alternation is not optimism.

\FloatBarrier
\bibliography{references}

@article{hart2000simple,
  title={A simple adaptive procedure leading to correlated equilibrium},
  author={Hart, Sergiu and Mas-Colell, Andreu},
  journal={Econometrica},
  volume={68},
  number={5},
  pages={1127--1150},
  year={2000},
  publisher={Wiley Online Library}
}

@inproceedings{farina2021faster,
  title={Faster game solving via predictive blackwell approachability: Connecting regret matching and mirror descent},
  author={Farina, Gabriele and Kroer, Christian and Sandholm, Tuomas},
  booktitle={Proceedings of the AAAI Conference on Artificial Intelligence},
  volume={35},
  number={6},
  pages={5363--5371},
  year={2021}
}

@inproceedings{farina2019stable,
  title={Stable-predictive optimistic counterfactual regret minimization},
  author={Farina, Gabriele and Kroer, Christian and Brown, Noam and Sandholm, Tuomas},
  booktitle={International conference on machine learning},
  pages={1853--1862},
  year={2019},
  organization={PMLR}
}

@article{zinkevich2007regret,
  title={Regret minimization in games with incomplete information},
  author={Zinkevich, Martin and Johanson, Michael and Bowling, Michael and Piccione, Carmelo},
  journal={Advances in neural information processing systems},
  volume={20},
  year={2007}
}

@article{tammelin2014solving,
  title={Solving large imperfect information games using {CFR+}},
  author={Tammelin, Oskari},
  journal={arXiv preprint arXiv:1407.5042},
  year={2014}
}

@article{bowling2015heads,
  title={Heads-up limit hold'em poker is solved},
  author={Bowling, Michael and Burch, Neil and Johanson, Michael and Tammelin, Oskari},
  journal={Science},
  volume={347},
  number={6218},
  pages={145--149},
  year={2015},
  publisher={American Association for the Advancement of Science}
}

@inproceedings{tammelin2015solving,
  title={Solving Heads-Up Limit Texas Hold'em.},
  author={Tammelin, Oskari and Burch, Neil and Johanson, Michael and Bowling, Michael},
  booktitle={IJCAI},
  volume={15},
  pages={645--652},
  year={2015}
}

@article{burch2019revisiting,
  title={Revisiting {CFR+} and alternating updates},
  author={Burch, Neil and Moravcik, Matej and Schmid, Martin},
  journal={Journal of Artificial Intelligence Research},
  volume={64},
  pages={429--443},
  year={2019}
}

@article{cevher2023alternation,
  title={Alternation makes the adversary weaker in two-player games},
  author={Cevher, Volkan and Cutkosky, Ashok and Kavis, Ali and Piliouras, Georgios and Skoulakis, Stratis and Viano, Luca},
  journal={Advances in Neural Information Processing Systems},
  volume={36},
  pages={18263--18290},
  year={2023}
}

@inproceedings{lee2024fundamental,
  title={Fundamental Benefit of Alternating Updates in Minimax Optimization},
  author={Lee, Jaewook and Cho, Hanseul and Yun, Chulhee},
  booktitle={International Conference on Machine Learning},
  pages={26439--26514},
  year={2024},
  organization={PMLR}
}

@inproceedings{nan2026on,
  title={On the {$O(1/T)$} Convergence of Alternating Gradient Descent{\textendash}Ascent in Bilinear Games},
  author={Tianlong Nan and Shuvomoy Das Gupta and Garud Iyengar and Christian Kroer},
  booktitle={The Fourteenth International Conference on Learning Representations},
  year={2026},
  url={https://openreview.net/forum?id=iBC8CGNGAS}
}

@inproceedings{zhang2022near,
  title={Near-optimal local convergence of alternating gradient descent-ascent for minimax optimization},
  author={Zhang, Guodong and Wang, Yuanhao and Lessard, Laurent and Grosse, Roger B},
  booktitle={International Conference on Artificial Intelligence and Statistics},
  pages={7659--7679},
  year={2022},
  organization={PMLR}
}

@article{syrgkanis2015fast,
  title={Fast convergence of regularized learning in games},
  author={Syrgkanis, Vasilis and Agarwal, Alekh and Luo, Haipeng and Schapire, Robert E},
  journal={Advances in Neural Information Processing Systems},
  volume={28},
  year={2015}
}

@inproceedings{rakhlin2013online,
  title={Online learning with predictable sequences},
  author={Rakhlin, Alexander and Sridharan, Karthik},
  booktitle={Conference on Learning Theory},
  pages={993--1019},
  year={2013},
  organization={PMLR}
}

@article{rakhlin2013optimization,
  title={Optimization, learning, and games with predictable sequences},
  author={Rakhlin, Sasha and Sridharan, Karthik},
  journal={Advances in Neural Information Processing Systems},
  volume={26},
  year={2013}
}

@article{farina2023regret,
  title={Regret matching{$^+$}: (In)stability and fast convergence in games},
  author={Farina, Gabriele and Grand-Cl{\'e}ment, Julien and Kroer, Christian and Lee, Chung-Wei and Luo, Haipeng},
  journal={Advances in Neural Information Processing Systems},
  volume={36},
  pages={61546--61572},
  year={2023}
}

@inproceedings{brown2019solving,
  title={Solving imperfect-information games via discounted regret minimization},
  author={Brown, Noam and Sandholm, Tuomas},
  booktitle={Proceedings of the AAAI Conference on Artificial Intelligence},
  volume={33},
  number={01},
  pages={1829--1836},
  year={2019}
}

@article{kuhn2016simplified,
  title={A simplified two-person poker},
  author={Kuhn, Harold W},
  journal={Contributions to the Theory of Games},
  volume={1},
  pages={97--103},
  year={1950}
}

@inproceedings{southey2005bayes,
  title={Bayes' bluff: opponent modelling in poker},
  author={Southey, Finnegan and Bowling, Michael and Larson, Bryce and Piccione, Carmelo and Burch, Neil and Billings, Darse and Rayner, Chris},
  booktitle={Proceedings of the Twenty-First Conference on Uncertainty in Artificial Intelligence},
  pages={550--558},
  year={2005}
}

@incollection{neller2011approximating,
  title={Approximating optimal Dudo play with fixed-strategy iteration counterfactual regret minimization},
  author={Neller, Todd W and Hnath, Steven},
  booktitle={Advances in Computer Games},
  pages={170--183},
  year={2011},
  publisher={Springer}
}

@inproceedings{hait2025alternating,
  title={Alternating Regret for Online Convex Optimization},
  author={Hait, Soumita and Li, Ping and Luo, Haipeng and Zhang, Mengxiao},
  booktitle={The Thirty Eighth Annual Conference on Learning Theory},
  pages={2632--2633},
  year={2025},
  organization={PMLR}
}

@article{wibisono2022alternating,
  title={Alternating mirror descent for constrained min-max games},
  author={Wibisono, Andre and Tao, Molei and Piliouras, Georgios},
  journal={Advances in Neural Information Processing Systems},
  volume={35},
  pages={35201--35212},
  year={2022}
}

@inproceedings{bailey2020finite,
  title={Finite regret and cycles with fixed step-size via alternating gradient descent-ascent},
  author={Bailey, James P and Gidel, Gauthier and Piliouras, Georgios},
  booktitle={Conference on Learning Theory},
  pages={391--407},
  year={2020},
  organization={PMLR}
}

@inproceedings{daskalakis2019last,
  title={Last-Iterate Convergence: Zero-Sum Games and Constrained Min-Max Optimization},
  author={Daskalakis, C and Panageas, Ioannis},
  booktitle={10th Innovations in Theoretical Computer Science (ITCS) conference, ITCS 2019},
  year={2019}
}

@article{lee2021last,
  title={Last-iterate convergence in extensive-form games},
  author={Lee, Chung-Wei and Kroer, Christian and Luo, Haipeng},
  journal={Advances in Neural Information Processing Systems},
  volume={34},
  pages={14293--14305},
  year={2021}
}

@inproceedings{cai2025lastiterate,
  title={Last-Iterate Convergence Properties of Regret-Matching Algorithms in Games},
  author={Yang Cai and Gabriele Farina and Julien Grand-Cl{\'e}ment and Christian Kroer and Chung-Wei Lee and Haipeng Luo and Weiqiang Zheng},
  booktitle={The Thirteenth International Conference on Learning Representations},
  year={2025},
  url={https://openreview.net/forum?id=LWeVVPuIx0}
}

@inproceedings{meng2025lastiterate,
  title={Last-Iterate Convergence of Smooth Regret Matching{$^+$} Variants in Learning {Nash} Equilibria},
  author={Linjian Meng and Youzhi Zhang and Zhenxing Ge and Tianyu Ding and Shangdong Yang and Zheng Xu and Wenbin Li and Yang Gao},
  booktitle={The Thirty-ninth Annual Conference on Neural Information Processing Systems},
  year={2025},
  url={https://openreview.net/forum?id=v771lscnlS}
}

@article{brown2018superhuman,
  title={Superhuman {AI} for heads-up no-limit poker: {Libratus} beats top professionals},
  author={Brown, Noam and Sandholm, Tuomas},
  journal={Science},
  volume={359},
  number={6374},
  pages={418--424},
  year={2018},
  publisher={American Association for the Advancement of Science}
}

@article{moravvcik2017deepstack,
  title={{DeepStack}: Expert-level artificial intelligence in heads-up no-limit poker},
  author={Morav{\v{c}}{\'\i}k, Matej and Schmid, Martin and Burch, Neil and Lis{\`y}, Viliam and Morrill, Dustin and Bard, Nolan and Davis, Trevor and Waugh, Kevin and Johanson, Michael and Bowling, Michael},
  journal={Science},
  volume={356},
  number={6337},
  pages={508--513},
  year={2017},
  publisher={American Association for the Advancement of Science}
}

@article{blackwell1956analog,
  title={An analog of the minimax theorem for vector payoffs},
  author={Blackwell, David},
  journal={Pacific Journal of Mathematics},
  volume={6},
  number={1},
  pages={1--8},
  year={1956},
  publisher={Mathematical Sciences Publishers}
}

@misc{tao2026optimalalternatingregretonline, title={Optimal Alternating Regret for Online Learning and Games}, author={Yixin Tao and Weiqiang Zheng}, year={2026}, eprint={2608.24731}, archivePrefix={arXiv}, primaryClass={cs.LG}, url={https://arxiv.org/abs/2608.24731}}

@misc{liu2026constantindividualregretgeneral,
      title={Constant Individual Regret in General Games},
      author={Mingyang Liu and Gabriele Farina and Asuman Ozdaglar},
      year={2026},
      eprint={2608.31166},
      archivePrefix={arXiv},
      primaryClass={cs.LG},
      url={https://arxiv.org/abs/2608.31166},
}

@inproceedings{xu2024dynamic,   title={Dynamic discounted counterfactual regret minimization},   author={Xu, Hang and Li, Kai and Fu, Haobo and Fu, Qiang and Xing, Junliang and Cheng, Jian},   booktitle={The Twelfth International Conference on Learning Representations},   year={2024} }

@inproceedings{li2025efficient,   title={Efficient online pruning and abstraction for imperfect information extensive-form games},   author={Li, Boning and Huang, Longbo},   booktitle={The Thirteenth International Conference on Learning Representations},   year={2025} }

@article{li2026real,   title={Real-Time Parallel Counterfactual Regret Minimization},   author={Li, Boning and Huang, Longbo},   journal={arXiv preprint arXiv:2605.19928},   year={2026} }

@inproceedings{ meng2026a, title={A Faster Parameter-Free Regret Matching Algorithm}, author={Linjian Meng and Youzhi Zhang and Shangdong Yang and Wenbin Li and Tianyu Ding and Yang Gao}, booktitle={The Fourteenth International Conference on Learning Representations}, year={2026}, url={https://openreview.net/forum?id=JLllvi7dsg} }

@article{zhang2025scale,   title={Scale-invariant regret matching and online learning with optimal convergence: Bridging theory and practice in zero-sum games},   author={Zhang, Brian Hu and Anagnostides, Ioannis and Sandholm, Tuomas},   journal={arXiv preprint arXiv:2510.04407},   year={2025} }

@article{li2026gpu, title={GPU-CFR: 80x Faster Counterfactual Regret Minimization by Compiling the Game to Static Dataflow and CUDA Graph Replay}, author={Li, Boning and Huang, Longbo}, journal={arXiv preprint arXiv:2609.11923}, year={2026}}

@article{lazarsfeld2026optimism,   title={Optimism Without Regularization: Constant Regret in Zero-Sum Games},   author={Lazarsfeld, John and Piliouras, Georgios and Sim, Ryann and Skoulakis, Stratis},   journal={Advances in Neural Information Processing Systems},   volume={38},   pages={121723--121758},   year={2026} }

@inproceedings{li2024rl,   title={RL-CFR: improving action abstraction for imperfect information extensive-form games with reinforcement learning},   author={Li, Boning and Fang, Zhixuan and Huang, Longbo},   booktitle={Proceedings of the 41st International Conference on Machine Learning},   pages={27752--27770},   year={2024} }

@article{li2026effective,   title={Effective, Efficient, and General Information Abstraction for Imperfect-Information Extensive-Form Games},   author={Li, Boning and Huang, Longbo},   journal={arXiv preprint arXiv:2605.10900},   year={2026} }

@article{li2026pokerskill,   title={PokerSkill: LLMs Can Play Expert-Level Poker without Training or Solvers},   author={Li, Boning and Wang, Baoxiang and Huang, Longbo},   journal={arXiv preprint arXiv:2605.30094},   year={2026} }

@inproceedings{li2026abstraction, title={Abstraction Agent: LLM-Guided Abstraction for Imperfect-Information Games}, author={Li, Boning and Huang, Longbo}, booktitle={Findings of the Empirical Methods in Natural Language Processing}, year={2026}}

@article{li2026correlated, title={Correlated Chance Sampling for Monte Carlo Counterfactual Regret Minimization}, author={Li, Boning and Chen, Yu and Huang, Longbo}, journal={arXiv preprint arXiv:2607.27035}, year={2026}}

@article{li2026agents, title={Agents That Certify Their Own Exploits: Confidence-Scheduled Restricted Responses for Safe Opponent Exploitation}, author={Li, Boning and Huang, Longbo}, journal={arXiv preprint arXiv:2607.28520}, year={2026}}

@article{li2026av, title={AV-AIVAT: 74x Cheaper Agent Evaluation with Certified Anytime-Valid Stopping in Imperfect-Information Games}, author={Li, Boning and Chen, Yu and Huang, Longbo}, journal={arXiv preprint arXiv:2608.06362}, year={2026}}

@article{li2026last,
  title={Why Last-Iterate Scale-Invariant Regret Matching Converges Linearly?},
  author={Li, Boning and Huang, Longbo},
  journal={arXiv preprint arXiv:2609.38162},
  year={2026}
}

@article{daskalakis2018limit,   title={The limit points of (optimistic) gradient descent in min-max optimization},   author={Daskalakis, Constantinos and Panageas, Ioannis},   journal={Advances in neural information processing systems},   volume={31},   year={2018} }

@article{farina2019optimistic,   title={Optimistic regret minimization for extensive-form games via dilated distance-generating functions},   author={Farina, Gabriele and Kroer, Christian and Sandholm, Tuomas},   journal={Advances in neural information processing systems},   volume={32},   year={2019} }

@inproceedings{wei2021linear, title={Linear Last-iterate Convergence in Constrained Saddle-point Optimization}, author={Chen-Yu Wei and Chung-Wei Lee and Mengxiao Zhang and Haipeng Luo}, booktitle={International Conference on Learning Representations}, year={2021}, url={https://openreview.net/forum?id=dx11_7vm5_r} }

@inproceedings{abe2023last,   title={Last-Iterate Convergence with Full and Noisy Feedback in Two-Player Zero-Sum Games},   author={Abe, Kenshi and Ariu, Kaito and Sakamoto, Mitsuki and Toyoshima, Kentaro and Iwasaki, Atsushi},   booktitle={International Conference on Artificial Intelligence and Statistics},   pages={7999--8028},   year={2023},   organization={PMLR} }

@article{cai2024fast,   title={Fast last-iterate convergence of learning in games requires forgetful algorithms},   author={Cai, Yang and Farina, Gabriele and Grand-Cl{\'e}ment, Julien and Kroer, Christian and Lee, Chung-Wei and Luo, Haipeng and Zheng, Weiqiang},   journal={Advances in Neural Information Processing Systems},   volume={37},   pages={23406--23434},   year={2024} }

@inproceedings{xu2024minimizing,   title={Minimizing weighted counterfactual regret with optimistic online mirror descent},   author={Xu, Hang and Li, Kai and Liu, Bingyun and Fu, Haobo and Fu, Qiang and Xing, Junliang and Cheng, Jian},   booktitle={Proceedings of the Thirty-Third International Joint Conference on Artificial Intelligence},   pages={5272--5280},   year={2024} }

@inproceedings{fiegel2025harder,   title={The Harder Path: Last Iterate Convergence for Uncoupled Learning in Zero-Sum Games with Bandit Feedback},   author={Fiegel, C{\^o}me and Menard, Pierre and Kozuno, Tadashi and Valko, Michal and Perchet, Vianney},   booktitle={Forty-second International Conference on Machine Learning},   year={2025} }

@inproceedings{soleymani2025faster,
  title={Faster rates for no-regret learning in general games via cautious optimism},
  author={Soleymani, Ashkan and Piliouras, Georgios and Farina, Gabriele},
  booktitle={Proceedings of the 57th Annual ACM Symposium on Theory of Computing},
  pages={518--529},
  year={2025}
}

@inproceedings{meng2026faster,   title={Faster Game Solving via Asymmetry of Step Sizes},   author={Meng, Linjian and Yang, Tianpei and Zhang, Youzhi and Ge, Zhenxing and Gao, Yang},   booktitle={Proceedings of the AAAI Conference on Artificial Intelligence},   volume={40},   number={20},   pages={17161--17169},   year={2026} }

@inproceedings{ICLR2026_b0b932e3,
 author = {Anagnostides, Ioannis and Tewolde, Emanuel and Zhang, Brian and Panageas, Ioannis and Conitzer, Vincent and Sandholm, Tuomas},
 booktitle = {International Conference on Learning Representations},
 editor = {C. Vondrick and B. Hariharan and C. Raffel and L. Pinto and D. Yang and A. Faust},
 pages = {108681--108707},
 title = {Convergence of Regret Matching in Potential Games and Constrained Optimization},
 url = {https://proceedings.iclr.cc/paper_files/paper/2026/file/b0b932e3b588a8e25c395433dbff769f-Paper-Conference.pdf},
 volume = {2026},
 year = {2026}
}

@misc{orabona2026lastiterate, title={Last-Iterate Convergence of Optimistic Multiplicative Weight Update}, author={Orabona, Francesco}, year={2026}, eprint={2606.11773}, archivePrefix={arXiv}, primaryClass={math.OC}, url={https://arxiv.org/abs/2606.11773} }

@article{li2026icm, title={ICM Out! Better Tournament Strategy from Computed Continuations, vs. Solvers and LLMs}, author={Li, Boning and Huang, Longbo}, journal={arXiv preprint arXiv:2608.09586}, year={2026}}

@misc{dewasurendra2026selfboundingregretmatchingpotential, title={Self-Bounding Regret Matching+ in Potential Games and Product-Simplex Optimization}, author={Pahan Dewasurendra and Subhashini Jayawardhana}, year={2026}, eprint={2608.17417}, archivePrefix={arXiv}, primaryClass={cs.GT}, url={https://arxiv.org/abs/2608.17417}}

@inproceedings{cai2025from, title={From Average-Iterate to Last-Iterate Convergence in Games: A Reduction and Its Applications}, author={Yang Cai and Haipeng Luo and Chen-Yu Wei and Weiqiang Zheng}, booktitle={The Thirty-ninth Annual Conference on Neural Information Processing Systems}, year={2025}, url={https://openreview.net/forum?id=mG9xItYI0D} }

@misc{savci2026gariprunningaveragemovingreference, title={GARIP: A Running-Average Moving Reference for Last-Iterate Self-Play in Two-Player Zero-Sum Games}, author={Can Savc{\i}}, year={2026}, eprint={2606.22688}, archivePrefix={arXiv}, primaryClass={cs.MA}, url={https://arxiv.org/abs/2606.22688} }

@article{cai2026separation, title={Separation of Nonergodic Uniform Convergence Rates for Regularized Learning in Games}, author={Cai, Yang and Farina, Gabriele and Grand-Cl{\'e}ment, Julien and Kroer, Christian and Lee, Chung-Wei and Luo, Haipeng and Zheng, Weiqiang}, journal={Operations Research}, year={2026}, publisher={INFORMS}}
\bibliographystyle{dilab_ref}
\makeappendixtoc

\appendix
\section{Proofs}
\label{app:proofs}

Throughout, strategies are parameterized by their first component and the witness game
is $A(a)$, $a \in (1,4]$, with $\delta^1(y) = (a+3)y - 2$, $\delta^2(x) = 2 - (a+3)x$,
and $x^{*} = y^{*} = 2/(a+3)$. All algebraic identities below are additionally verified
by symbolic computation (15/15 checks; \cref{tab:sympy}).

\subsection{Proof of \cref{lem:sign} (sign structure)}

For player 2, $u^2(\vx) = -A^{\top}\vx$, so $u^2_1 - u^2_2 = \delta^2(x)$ and
\[
\begin{aligned}
g^2_1 &= u^2_1 - \big(y\,u^2_1 + (1-y)\,u^2_2\big) = (1-y)\,\delta^2(x),\\
g^2_2 &= u^2_2 - \big(y\,u^2_1 + (1-y)\,u^2_2\big) = -y\,\delta^2(x).
\end{aligned}
\]
The formula for $g^1$ is identical with $\delta^1(y)$, and the sign statements follow by
inspection of the products. \hfill$\square$

\subsection{Proof of \cref{lem:readout} (exact-vertex characterization)}

By \cref{lem:sign}, $\vm = g^2(\vx_{t+1}, \vy_t) = ((1-y_t)\Delta, -y_t\Delta)$ with
$\Delta = \delta^2(x_{t+1})$, and the readout is
$\vy_{t+1} = \sigma\big([\alpha + (1-y_t)\Delta]^{+}, [\beta - y_t\Delta]^{+}\big)$.
It equals $e_1$ exactly when the second entry truncates to zero ($\beta \le y_t\Delta$)
and the first entry is positive; it equals $e_2$ exactly in the symmetric case
$\alpha \le -(1-y_t)\Delta$ with $\beta - y_t\Delta > 0$. If both entries truncate,
the readout vector is $\mathbf{0}$ and the convention applies. Otherwise at least one
entry is positive; the readout is a vertex in the two cases above, uniform when
$[\alpha+(1-y_t)\Delta]^{+} = [\beta-y_t\Delta]^{+} > 0$ (for instance $\alpha =
\beta$ and $\Delta = 0$), and interior otherwise. In
particular a zero disfavored coordinate ($\beta = 0$, $\Delta > 0$) satisfies the first
inequality for every positive signal with $y_t < 1$. \hfill$\square$

\subsection{Robust variant (\cref{lem:readout}$'$, $\rho$-near-vertex readout)}
\label{app:lem2p}

\begin{lemma}[$\rho$-near-vertex readout]
\label{lem:readout-robust}
Suppose $\max(\alpha, \beta) \le \bar\varepsilon$, $\Delta \ge \Delta_0 > 0$, and
$y_t \le 1 - \theta$ for some $\theta > 0$. Then
\[
y_{t+1} \;\ge\; \frac{\theta\Delta_0}{\theta\Delta_0 + \bar\varepsilon},
\qquad\text{i.e.}\qquad
\|\vy_{t+1} - e_1\|_1 \le 2\rho, \quad
\rho := \frac{\bar\varepsilon}{\theta\Delta_0 + \bar\varepsilon},
\]
and symmetrically for $\Delta \le -\Delta_0$, $y_t \ge \theta$. Consequently, if
domination holds only approximately (up to an $\bar\varepsilon$-sized violation), the
separation constant in \cref{thm:t1} degrades from $2\eta$ to $2(\eta - \rho)$.
\end{lemma}

\begin{proof}
The favored component of $[\vq^{2+} + \vm]^{+}$ is at least $(1-y_t)\Delta \ge
\theta\Delta_0$; the disfavored component is at most $\beta \le \bar\varepsilon$. The
readout $y_{t+1} = L/(L+R)$ is monotone in both, giving the bound.
\end{proof}

The lemma is useful only when $\rho \ll \eta$, i.e.\ when the persistent state is small
\emph{relative to the favored signal mass} $\theta\Delta_0$. On the witness this is not the
case: at jump times the ratio $\varepsilon_t/(\theta_t\Delta_t + \varepsilon_t)$ has median
$0.21$ and reaches $1$ (\cref{tab:prec}), because the mechanism operates with signals of the
same order as the state. The lemma therefore gives no floating-point margin. The robustness
evidence is the four-precision agreement of every certificate statistic.

\subsection{Proof of \cref{lem:washout} (best-response-matched suppression)}

With $\vy_t = e_1$, \cref{lem:sign} gives $g^2(\vx_t, e_1) = (0, -\delta^2(x_t))$.
If $\delta^2(x_t) \ge 0$ the update adds a nonpositive vector, and the state
decreases componentwise, $\vq^{2+}_{t+1} = (q^2_L, [q^2_R - \delta^2(x_t)]^{+}) \le \vq^2_t$.
The $e_2$ case is symmetric. \hfill$\square$

The mechanism is self-reinforcing: vertex play aligned with the stale signal washes the
disfavored persistent component down (the favored one is unchanged); a small persistent
state makes the next domination easier
(\cref{lem:readout}). Empirically the tail persistent state satisfies
$\|\vq^{2,\mathrm{OS}}\|_\infty \le 8.5\times 10^{-3}$, versus component sums of
$37$--$60$ for \ALT's $\vq^2$ on the same games.

\subsection{Proof of \cref{lem:classify} (vertex--uniform classification)}

If $\vq^{2+}_{t+1} = \mathbf{0}$, then $[\vq^{2+} + \vm]^{+} = [\vm]^{+}$, and by
\cref{lem:sign} the vector $\vm$ has opposite-signed components (or a zero component at
the boundary), so $[\vm]^{+}$ is either a single-coordinate positive vector (readout:
exact vertex, selected by $\mathrm{sign}\,\Delta$) or zero (readout: uniform, occurring
when $\Delta = 0$ or when the sign points into the boundary $y_t \in \{0,1\}$ from
outside). \hfill$\square$

\subsection{Proof of \cref{lem:complex,lem:growth} (complex update and normalizer
growth)}

Write $z = q_L + \mathrm{i}q_R$ and $\varepsilon = \delta/S$. For the two-action
increment $\vg = \delta(1-x, -x)$, $g_L = (1-x)\delta = (q_R/S)\delta = \varepsilon
q_R$ and $g_R = -x\delta = -\varepsilon q_L$, so $z^{+} = (q_L + \varepsilon q_R) +
\mathrm{i}(q_R - \varepsilon q_L) = z(1 - \mathrm{i}\varepsilon)$. Then
$|z^{+}|^{2} = |z|^{2}(1+\varepsilon^{2})$ and $|z|^{2}\varepsilon^{2} =
\delta^{2}|z|^{2}/S^{2} = \delta^{2}(x^{2} + (1-x)^{2})$ because $q_L = xS$ and
$q_R = (1-x)S$, giving \eqref{eq:complex}; both components of $\vq + \vg$ are positive
exactly when $z^{+}$ lies in the open first quadrant. Monotonicity of $x$ in $\alpha$
follows from its definition, and $|\arctan\varepsilon| \le |\varepsilon|$.

On a truncation-free window, $q^1_{1,t} = x_t S^1_t$ and the update gives
\[
x_{t+1} - x_t
= \frac{g^1_1 - x_t\,\Sigma g^1}{S^1_{t+1}},
\qquad
\Sigma g^1 = (1 - 2x_t)\,\delta^1(y_t),
\]
and by \cref{lem:sign},
$g^1_1 - x_t \Sigma g^1 = \delta^1(y_t)\,[(1-x_t) - x_t(1 - 2x_t)] =
\delta^1(y_t)\,\varphi(x_t)$ with $\varphi(u) = u^{2} + (1-u)^{2} = 2u^2 - 2u + 1$
attaining its minimum $\nicefrac12$ at $u = \nicefrac12$; the $y$-update is identical
with $\delta^2(x_{t+1})$ and $S^2$. The sign statements follow since $\varphi > 0$ and
$S^i > 0$, and $\delta^1(y) \gtrless 0$ iff $y \gtrless y^{*}$,
$\delta^2(x) \gtrless 0$ iff $x \lessgtr x^{*}$.

Both \ALT{} states are nonzero from the first update on. From zero initialization
player~1 updates first against the uniform $\vy_1$ with $\delta^1(\nicefrac12) =
(a-1)/2 > 0$, so $\vq^1_2 = ((a-1)/4,\,0) \neq \mathbf{0}$; player~2 then updates
against $\vx_2 = e_1$ with $\delta^2(1) = -(a+1) < 0$, so $\vq^2_2 = (0,\,(a+1)/2)
\neq \mathbf{0}$. A nonzero state never returns to zero: $\vq_{t+1} = \mathbf{0}$
requires $\vg_t \le -\vq_t$ componentwise, while $\langle \vq_t, \vg_t\rangle =
S_t\,\vx_t^{\top}\vg_t = 0$ forces $g_{t,i} = 0$ on the support of $\vq_t$, where
then $q_{t,i} + g_{t,i} = q_{t,i} > 0$. Hence at a nonzero state, interior play
$x_{t+1} \in (0,1)$ means both coordinates of $\vq^1_{t+1}$ are positive, i.e.\ round
$t$ is truncation-free; \cref{ass:c0} therefore covers exactly the rounds $t \ge T_0$
summed in \eqref{eq:growth}.

Under \cref{ass:c0}, iterating \eqref{eq:complex} over the truncation-free rounds gives
\eqref{eq:growth}, and $S = |z|(\cos\alpha + \sin\alpha) \ge |z|$ on the first
quadrant. Since $\varphi \ge \nicefrac12$, $|z^1_T|^{2} \ge \tfrac12 \sum_{t=T_0}^{T-1}
\delta^{1}(y_t)^{2}$, so a divergent sum forces $S^{1}_t \to \infty$ and the angular
step $\arctan(|\delta^{1}(y_t)|/S^{1}_t) \to 0$; if instead the sum converges then
$\delta^{1}(y_t) \to 0$, and $\delta^{1}$ is affine with nonzero slope, so $y_t \to
y^{*}$. The claims for player~2 are symmetric. \hfill$\square$

\subsection{Proof of \cref{prop:epoch} (epoch structure)}

(i) The persistent increment is $g^{2}(\vx_t, \vy_t) = (1-y_t, -y_t)\,\delta^{2}(x_t)$
by \cref{lem:sign}, with $\delta^{2}(x_t) > 0$ on $E$; the first
coordinate never truncates (it only grows), the second decreases until it truncates to
zero and then stays zero as long as the increment is nonpositive, i.e.\ throughout $E$.
(ii) With $\beta = q^{2}_{R,t+1} = 0$, $\alpha = q^{2}_{L,t+1} > 0$ and $\Delta =
\delta^{2}(x_{t+1}) > 0$, \cref{lem:readout} gives $\vy_{t+1} = e_1$. At $y_{t+1} = 1$ the
next persistent increment is $(0, -\delta^{2}(x_{t+1}))$, so $q^{2}_{R,t+2} = 0$ and
$q^{2}_{L,t+2} = \alpha > 0$; induction along $E$ repeats the argument while $\Delta > 0$.
(iii) By (ii) the run $\vy_u = e_1$ covers $\{t+1, \dots, e\}$; if $E$ were infinite
the run would be infinite. Against $\vy = e_1$ player~1 receives
$g^{1}(\vx, e_1) = (a+1)(1-x, -x)$, so $q^{1}_L$ increases by $(a+1)(1-x_u) > 0$
whenever $x_u < 1$ and $q^{1}_R$ is nonincreasing; hence $x_u = q^{1}_L/(q^{1}_L +
q^{1}_R)$ is strictly increasing while $x_u < 1$. If $x_u \le 1-c$ for all $u$ then
$q^{1}_L \to \infty$ with $q^{1}_R$ bounded, forcing $x_u \to 1$, a contradiction; so
$x_u \to 1 > x^{*}$ (or $x_u = 1$ at some round), and $\delta^{2}(x_u) < 0$ eventually,
contradicting $\delta^2 > 0$ on $E$.
The next epoch is negative. By maximality of $E$, $\delta^{2}(x_{e+1}) \le 0$. If it is
negative, the next epoch starts at $e+1$. If $\delta^{2}(x_{e+1}) = 0$, then $x_{e+1}
= x^{*}$, and the readout at round $e$ has $\Delta = 0$ with $\beta = 0 < \alpha$, so
$\vy_{e+1} = e_1$ by \cref{lem:readout}; player~1 updates against $e_1$ once more and
$x_{e+2} > x_{e+1} = x^{*}$, so $\delta^{2}(x_{e+2}) < 0$ and the next epoch starts at
$e+2$, the round $e+1$ belonging to no epoch.
(iv) For an epoch of length at least two, by (i) the criterion produces an exactly zero
disfavored coordinate inside $E$ before
its last round; by (ii) the epoch then contains a vertex time; infinitely many such epochs
give infinitely many vertex times, which is branch~(a). \hfill$\square$

\subsection{Proof of \cref{lem:flip} (sign-flip readout)}
\label{app:lem-flip}

\begin{lemma}[Sign-flip readout]
\label{lem:flip}
\emergencystretch=2em
In $\Phi_{\OSPRMmath}$, let round $s$ satisfy $\delta^2(x_s) > 0$ and
$\delta^2(x_{s+1}) < 0$ (round $s$ is the last round of a positive epoch; the epoch has
length one iff also $\delta^2(x_{s-1}) \le 0$). Then
$\vy_{s+1} = e_2$ if and only if
\[
\begin{aligned}
&q^2_{L,s} + (1-y_s)\,\delta^2(x_s) \;\le\; (1-y_s)\,|\delta^2(x_{s+1})| \\
&\quad\text{and}\quad \big(y_s > 0 \ \text{ or } \ q^2_{R,s} > y_s\,\delta^2(x_s)\big),
\end{aligned}
\]
so a vertex at a sign flip requires the coordinate favored by the finished epoch to be at
most the growth of the signal magnitude across the flip,
$q^2_{L,s} \le (1-y_s)\big(|\delta^2(x_{s+1})| - \delta^2(x_s)\big)$; an entering zero
disfavored coordinate $q^2_{R,s} = 0$ is neither necessary nor sufficient. The symmetric
statement holds for a negative epoch of length one with $L$ and $R$ exchanged.
\end{lemma}
\begin{proof}
The persistent update at round $s$ adds $(1-y_s)\delta^2(x_s) \ge 0$ to $q^2_L$ (no
truncation) and gives $\beta = [q^2_{R,s} - y_s\,\delta^2(x_s)]^{+}$. The readout uses
$\Delta = \delta^2(x_{s+1}) < 0$, and \cref{lem:readout} gives $\vy_{s+1} = e_2$ iff
$\alpha \le -(1-y_s)\Delta = (1-y_s)|\Delta|$ and
$\beta - y_s\Delta = \beta + y_s|\Delta| > 0$, with
$\alpha = q^2_{L,s} + (1-y_s)\delta^2(x_s)$. The second condition holds iff $y_s > 0$ or
$\beta > 0$. The displayed consequence is the first condition rearranged; the last clause
follows because $q^2_{R,s}$ does not appear in it. \hfill$\square$
\end{proof}

\subsection{Proof of \cref{thm:t1}}

If (b) fails there are infinitely many vertex times, i.e.\
rounds with $\vy^{\mathrm{OS}}_{t+1} \in \{e_1, e_2\}$ exactly. At each such $t \ge T_0$,
\cref{ass:c1} places $y^{\ALTmath}_{t+1} \in [\eta, 1-\eta]$ (the $x$-part of
\cref{ass:c1} only secures \cref{ass:c0}), and for a vertex
$e \in \{e_1, e_2\}$,
$\|\vy^{\ALTmath}_{t+1} - e\|_1 = 2|y^{\ALTmath}_{t+1} - e_{(1)}| \ge 2\eta$ for both
choices of $e$, and $D_{t+1}$ dominates the $y$-block of the profile distance by
definition. \hfill$\square$

\subsection{Machine-checked identities}

\begin{table}[H]
\centering
\caption{The 15 symbolically verified identities. Each check treats
$(a, x, y, \alpha, \beta)$ as free symbols.}
\label{tab:sympy}
\small
\begin{tabularx}{\linewidth}{@{}l X@{}}
\toprule
\textbf{Group} & \textbf{Checks}\\
\midrule
\cref{lem:sign} & $g^1, g^2$ componentwise decomposition; opposite-sign structure (3)\\
\cref{lem:readout} & readout components under domination, both signs; zero case (3)\\
\cref{lem:readout-robust} & the $\rho$ identity and monotonicity bound (2)\\
\cref{lem:complex} & update identity $g^1_1 - x\Sigma g^1 = \delta^1\varphi(x)$;
$S$ recursion; $\varphi \in [\nicefrac12, 1]$ (4)\\
\cref{thm:t1} & vertex-distance inequality at both vertices; triangle structure of $D_t$ (3)\\
\bottomrule
\end{tabularx}
\end{table}

% ---------------------------------------------------------------------------
% A2_proofs_ext.tex — 新增附录: 5b/5c 的完整证明。\section 级, 放在 A_proofs.tex 之后。
% 2026-10-03 页数裁剪: 5c 的 prop:closed / prop:order / cor:channels 陈述
% 连同证明一并移入本文件 (正文只留摘要句), 故改名为 Statements and Proofs。
% ---------------------------------------------------------------------------
\section{Derivations}
\label{app:proofs-ext}

Throughout, $\|\cdot\|_1$ is the $\ell_1$ norm, $[\cdot]^{+}$ the componentwise
positive part, $\sigma(\vq) = \vq/\|\vq\|_1$ with $\sigma(\mathbf{0})$ uniform.

\begin{proof}[Proof of \cref{lem:norm-lip}]
If $\vu = \mathbf{0}$, then $\sigma(\vu)$ is uniform and
$\|\sigma(\vu) - \sigma(\vv)\|_1 \le 2$, the bound being $2\|\vu - \vv\|_1/\|\vv\|_1$.
Otherwise
\[
\begin{aligned}
\sigma(\vu) - \sigma(\vv) &= \frac{\vu - \vv}{\|\vv\|_1}
+ \vu\Big(\frac{1}{\|\vu\|_1} - \frac{1}{\|\vv\|_1}\Big),\\
\Big\|\vu\Big(\frac{1}{\|\vu\|_1} - \frac{1}{\|\vv\|_1}\Big)\Big\|_1
&= \frac{\big|\,\|\vv\|_1 - \|\vu\|_1\big|}{\|\vv\|_1}
\le \frac{\|\vu-\vv\|_1}{\|\vv\|_1},
\end{aligned}
\]
using $\vu \ge \mathbf{0}$ (so $\|\vu\|_1 = \sum_i u_i$) and the reverse triangle
inequality. Summing the two terms gives the claim.
\end{proof}

\begin{proof}[Proof of \cref{lem:g1}]
(i) $U(\mathbf{0}, \vm) = [\vm]^{+}$ and $[\lambda\vm]^{+} = \lambda[\vm]^{+}$ for
$\lambda > 0$; $\sigma$ is $0$-homogeneous. (ii) The map $z \mapsto [z]^{+}$ is
$1$-Lipschitz coordinatewise, so $\|[\vq + \vm]^{+} - [\vm]^{+}\|_1 \le \|\vq\|_1$.
Apply \cref{lem:norm-lip} with $\vu = [\vq+\vm]^{+}$, $\vv = [\vm]^{+}$.
\end{proof}

\begin{proof}[Proof of \cref{lem:g2}]
For $q_i \ge 0$: if $q_i + g_i \ge 0$ then $[q_i + g_i]^{+} - q_i = g_i$; otherwise
$[q_i+g_i]^{+} - q_i = -q_i$ and $q_i < -g_i = |g_i|$. Hence $\|[\vq+\vg]^{+} -
\vq\|_1 \le \|\vg\|_1$, and \cref{lem:norm-lip} with $\vv = \vq$ gives the bound.
\end{proof}

\begin{proof}[Proof of \cref{thm:gain}]
Every instantaneous regret vector has the form $\vg = \vu - \langle \vp, \vu\rangle
\mathbf{1}$ with $\vp$ in the simplex and $|u_i| \le \|A\|_{\max}$, so $\|\vg\|_1 \le
2n\|A\|_{\max}$. Part (a) (which assumes $\vq^2_t \ne 0$) is \cref{lem:g2} with
$\vg = \vf_t$; part (b) (which assumes $[\vf_t]_+ \ne 0$) is \cref{lem:g1}(ii) with
$\vm = \vf_t$, and the
$n=2$ statement is \cref{lem:sign}: $[\vf]^{+}$ has exactly one nonzero coordinate
when the two coordinates of $\vf$ have opposite signs.
\end{proof}

% R3-fill(F3a): cor:lastiter 陈述 已移入正文 (见 r3_pairs_fills.py)。

% R4-trim(F3b): cor:lastiter 证明 自正文移回 (见 r4_trim2.py)。

\begin{proof}[Proof of \cref{cor:lastiter}]
At a vertex time, $\vy^{\mathrm{OS}}_{t+1} \in \{e_1, e_2\}$ by definition.
The value of $A(a)$ is $v(a) = u(\vx^{*},\vy^{*})$; with $y^{*} = 2/(a+3)$,
$A(a)\vy^{*} = \big(\tfrac{a-1}{a+3}, \tfrac{a-1}{a+3}\big)$, so $v(a) =
\tfrac{a-1}{a+3}$. Against $e_1$ player~1's best response earns $\max(a, -1) = a$,
against $e_2$ it earns $\max(-1, 1) = 1$. Thus the exploitability of $e_1$ is $a -
v(a) = (a+1)^2/(a+3)$ and that of $e_2$ is $1 - v(a) = 4/(a+3)$; since $(a+1)^2 \ge 4$
for $a \ge 1$ the minimum is $4/(a+3)$, which is decreasing in $a$ and equals $4/7$ at
$a = 4$.
\end{proof}

\begin{proof}[Proof of \cref{rem:shift}]
$g^2(\vx,\vy) = -A^{\top}\vx + u(\vx,\vy)\mathbf{1}$, so
\[
\begin{aligned}
g^2(\vx_{t+1},\vy_t) - g^2(\vx_{t+1},\vy_{t+1})
&= [u(\vx_{t+1},\vy_t) - u(\vx_{t+1},\vy_{t+1})]\mathbf{1}.
\end{aligned}
\]
On the witness,
\[
|u(\vx_{t+1},\vy_t)-u(\vx_{t+1},\vy_{t+1})| = |y_t-y_{t+1}|\,|\delta^2(x_{t+1})|,
\]
so a macroscopic jump in $y$ carries a small regret-prediction error whenever
$x_{t+1}$ is near $x^*$ (measured p99 $\le 0.026$).
\end{proof}

\begin{proof}[Proof of \cref{prop:identity}]
Since $u$ is bilinear,
\begin{align*}
W\max_{\vx} u(\vx,\bar\vy)
&= \textstyle\max_{\vx} \sum_t w_t u(\vx,\vy_t)\\
&= R^1 + \textstyle\sum_t w_t u(\vx_t,\vy_t),\\
-W\min_{\vy} u(\bar\vx^{+},\vy)
&= \textstyle\max_{\vy} \sum_t w_t[-u(\vx_{t+1},\vy)]\\
&= R^2_{\mathrm{fresh}} - \textstyle\sum_t w_t u(\vx_{t+1},\vy_t).
\end{align*}
Adding,
\[
\begin{aligned}
W\,\mathrm{gap}(\bar\vx^{+},\bar\vy) &= R^1 + R^2_{\mathrm{fresh}} - \sum_t
w_t[u(\vx_{t+1},\vy_t) - u(\vx_t,\vy_t)]\\
&= R^1 + R^2_{\mathrm{fresh}} - \Delta.
\end{aligned}
\]
The
second identity is the same computation with $\vx_t$ in place of $\vx_{t+1}$, for which
the cross terms cancel exactly.
\end{proof}

\begin{proof}[Proof of \cref{prop:closed}]
$\vx_{t+1}$ and $\vx_t$ sum to one, so $(\vx_{t+1} - \vx_t)^{\top} A\vy_t =
(\vx_{t+1} - \vx_t)^{\top}\vg_t$. By definition of $\vg_t$, $\vx_t^{\top}\vg_t =
u(\vx_t,\vy_t) - u(\vx_t,\vy_t) = 0$, hence also $\langle \vq_t, \vg_t\rangle =
\|\vq_t\|_1\,\vx_t^{\top}\vg_t = 0$ (trivially if $\vq_t = \mathbf{0}$). If $S_{t+1} >
0$ then $\vx_{t+1} = \vq_{t+1}/S_{t+1}$ and
\[
(\vx_{t+1} - \vx_t)^{\top}\vg_t = \frac{\langle \vq_{t+1}, \vg_t\rangle}{S_{t+1}} =
\frac{\langle \vq_{t+1} - \vq_t, \vg_t\rangle}{S_{t+1}}.
\]
Coordinatewise, for $i \in \mathcal{A}_t$: $(q_{t,i} + g_{t,i} - q_{t,i})\,g_{t,i} =
g_{t,i}^2$; for $i \notin \mathcal{A}_t$: $(0 - q_{t,i})\,g_{t,i} = q_{t,i}|g_{t,i}|$
because $g_{t,i} < -q_{t,i} \le 0$. Both are nonnegative. If $S_{t+1} = 0$ then
$g_{t,i} \le -q_{t,i} \le 0$ for all $i$; since $\vx_t^{\top}\vg_t = 0$ with $\vx_t \ge
\mathbf{0}$ supported on $\{i : q_{t,i} > 0\}$ (or on all of $[m]$ if $\vq_t =
\mathbf{0}$), $g_{t,i} = 0$ on that support. If $\vq_t \ne \mathbf{0}$ this gives
$q_{t,i} + g_{t,i} = q_{t,i} > 0$ on the support, contradicting $S_{t+1} = 0$; so
$\vq_t = \mathbf{0}$, $\vx_t$ is uniform, $\vg_t \le \mathbf{0}$ has zero mean, hence
$\vg_t = \mathbf{0}$ and $\vx_{t+1} = \vx_t$, so the left-hand side is $0$.
\end{proof}

% R3-fill(F2a): prop:order 陈述 已移入正文 (见 r3_pairs_fills.py)。

% R4-trim(F2b): prop:order 证明 自正文移回 (见 r4_trim.py)。

\begin{proof}[Proof of \cref{prop:order}]
\emph{Potential.} $\|\vq_{t+1}\|_2^2 = \|[\vq_t + \vg_t]^{+}\|_2^2 \le \|\vq_t +
\vg_t\|_2^2 = \|\vq_t\|_2^2 + 2\langle \vq_t,\vg_t\rangle + \|\vg_t\|_2^2 =
\|\vq_t\|_2^2 + \|\vg_t\|_2^2$, using $\langle\vq_t,\vg_t\rangle = 0$ from the previous
proof; with $\vq_1 = \mathbf{0}$, induction gives $\|\vq_{t+1}\|_2^2 \le G_{t+1}^2$.
\emph{Lower bound.} By \cref{prop:closed}, each summand of $\Delta_T$ is at least
$\|\vg_t|_{\mathcal{A}_t}\|_2^2 / S_{t+1}$, and $S_{t+1} = \|\vq_{t+1}\|_1 \le
\sqrt{m}\,\|\vq_{t+1}\|_2 \le \sqrt{m}\,G_{t+1}$ (if $G_{t+1} = 0$ then $\vg_t =
\mathbf{0}$ and the summand is $0$). \emph{Telescoping.} If $\mathcal{A}_t = [m]$ then
$\|\vg_t\|_2^2 = G_{t+1}^2 - G_t^2 = (G_{t+1} - G_t)(G_{t+1} + G_t)$ and $(G_{t+1} +
G_t)/G_{t+1} \ge 1$, so $\|\vg_t\|_2^2/G_{t+1} \ge G_{t+1} - G_t$; sum over the window.
\emph{Regret.} $q_{t+1,i} = [q_{t,i} + g_{t,i}]^{+} \ge q_{t,i} + g_{t,i}$, so
$q_{T+1,i} \ge \sum_{t \le T} g_{t,i}$ and $R^1_T = \max_i \sum_t g_{t,i} \le \max_i
q_{T+1,i} \le \|\vq_{T+1}\|_2 \le G_{T+1}$.
\end{proof}

\begin{proof}[Proof of \cref{cor:channels}]
Insert both identities of \cref{eq:identity} into the log-ratio and cancel
$R^1 + R^2_{\mathrm{fresh}}$.
\end{proof}

% R4-trim(G1): machine-check 披露 自正文移回 (见 r4_trim2.py)。

Identities \eqref{eq:identity} and \eqref{eq:closed} and the bound \eqref{eq:order} are
verified numerically on every suite trajectory whose player~1 is plain \RMp{}
(\cref{app:protocol} lists the eligible arms; residual below $10^{-9}$ relative to the
largest term, per-step closed form to $10^{-10}$, and the lower bound, checked against
$\Delta$ recomputed with uniform weights, holds on all $1271$ eligible matrix-suite
trajectories), and
\cref{lem:g2} is verified step by step on \ALT{} (measured jump divided by the bound
never exceeds $1$).

\section{Registration}
\label{app:ladder}

This appendix documents the pre-registered claim ladder, the empirical failure of
its top rung, and the three proof routes attempted before the dichotomy theorem closed.
We include it because negative-result papers are only as credible as their disclosure
discipline.

\subsection{The ladder and its gates}

The ladder of \cref{sec:setup} was frozen before any theorem-directed experiment ran,
together with operational gates: the strong-separation rung (ii) required, on the
60-point witness scan, a nonempty parameter interval where
$\min_t D_t$ over the tail half stays $\ge 0.2$ while the windowed mean does not decay.
The frozen pre-registration file and the dated experiment tracker will be released
with the code upon acceptance.

\paragraph{Gate outcome (rung (ii) fails).}
The scan ($a \in [1.05, 4]$, 60 points, $T = 2\times 10^4$) produced: windowed means of
$D_t$ flat at $0.60$--$1.06$ on all 60 points (no decay anywhere, fourth-quarter vs
third-quarter window ratio $\ge 0.8$ on 60/60), but tail minima with median $0.012$ and
minimum $3.4\times 10^{-4}$; only two isolated grid points ($a = 3.15, 3.65$) met the
$\ge 0.2$ criterion, and no two adjacent grid points did. Deep runs at $T = 10^6$ pushed
tail minima lower still ($\sim 4\times 10^{-4}$). Mechanism (measured): \ALT's
rotation periodically sweeps past \OSPRM's vertex--uniform jump set, so on every tested
horizon the tail minimum of $D_t$ falls well below $0.2$ even though the process spends most of
its time at $D_t = \Theta(1)$. The gate verdict was FAIL; per pre-registration we
downgraded the proof target to weak trajectory separation (rung (iii)) and report the
mean-level flatness as empirical context only.

\subsection{Three proof routes, honestly accounted}

\paragraph{Route 1. Time-averaged separation.}
Bound $\bar D_T \ge |\bar y^{\ALTmath} - \bar y^{\mathrm{OS}}|$ and show the \OSPRM{}
average stays away from $y^{*}$. Excluded by our own measurement: the \OSPRM{}
\emph{duality gap of averages} also decays ($\sim T^{-1/2}$; $4$--$7 \times 10^{-4}$
at $T = 10^6$), so both average gaps decay on the measured horizon and the route cannot yield
a constant. Contribution retained: none (the route is dead), but the measurement
sharpened the claim that the separation is a trajectory phenomenon, invisible to
averages.

\paragraph{Route 2. Direct cascade by contradiction.}
Assume bounded $\vq^2$ and no vertex visits; derive that $x^{\mathrm{OS}}$ sticks near
$x^{*}$; contradict the player-1 regret dynamics. The final step (excluding
\emph{asymptotic} sticking; exact sticking is excludable) turned
out to require proving non-convergence of the \OSPRM{} last iterate, which we
could not close elementarily. Contributions retained: \cref{lem:readout},
\cref{lem:readout-robust}, \cref{lem:washout}, and the dichotomy structure of
\cref{thm:t1}, whose branch (b) is exactly the un-closable step, now isolated and
made the object of a direct test.

\paragraph{Route 3. ODE / stochastic-approximation limit.}
\cref{lem:complex} exhibits \ALT{} on truncation-free windows as a rotation field with
step $1/S_t$, inviting a Benaïm-style limiting-ODE analysis (and an $a = 2$
specialization). Establishing the required growth rate of $S_t$ is itself open, and the
program is a paper-sized undertaking that still would not touch the \OSPRM{}-side
difficulty. Contribution retained: \cref{lem:complex,lem:growth} as the mechanism lemmas behind
\cref{ass:c1}.

\paragraph{Final placement.}
The result landed between the pre-registered rungs: a rigorous dichotomy theorem with
one finite-horizon verified assumption and an explicit, empirically tight constant. This
is stronger than the pre-registered fallback for total proof failure (formalized
conjecture plus witness family), and weaker than an unconditional rung-(iii) theorem;
\cref{conj:cinf} records the remaining gap precisely.

\subsection{Exact-arithmetic validation}

As a check that the early dynamics are not floating-point artifacts (the four-precision
recomputation of \cref{tab:prec} covers long horizons statistically), the
first iterations of both dynamics were replayed in exact rational arithmetic at five
representative values of $a$. The rational and float trajectories agree to $10^{-12}$
for as long as the rational replay is feasible. Feasibility ends quickly, since
denominator bit-lengths roughly double per iteration (reaching $10^{4.5}$--$10^{5}$
bits by $t = 9$--$10$), so the exact replay covers the early iterations only.
Symbolic itineraries (truncation-activity patterns of both dynamics)
over long floating-point runs use at most 8 distinct symbols, with \ALT{} eventually
truncation-free (consistent with \cref{lem:complex}'s regime) and \OSPRM{} showing
persistent-state zeroing in $88$--$100\%$ of tail windows plus exact-vertex visits at
frequency $51$--$66\%$ on these five points (\cref{app:tables}).

\subsection{Scope of the claims}

Our claims are scoped to deterministic self-play dynamics of the tested update rules,
not to sampled (Monte Carlo) variants, discounted weighting \citep{brown2019solving},
or adversarial training loops. The theorem is for a $2\times 2$ family, and the
extensive-form evidence is witness-level, with one pre-registered gate failure
(\cref{sec:exp-efg}). Leduc-scale trees are far from the game sizes where alternation's
value was first observed \citep{bowling2015heads}. The catalog covers the dominant live
routes, not all conceivable explanations, and the staleness and damping dials are tested
down to $s=\lambda=10^{-2}$ and $10^{-3}$ only. Tree distances are per information set,
not reach-weighted.

Excluding branch (b) would require proving that \OSPRM{} plays exact vertices infinitely
often, a property stronger than last-iterate non-convergence and adjacent to the known
instability of \PRM-family last iterates \citep{farina2023regret,cai2025lastiterate}.
We offer no positive theory. The most natural candidate, play-level proximal damping, reproduces \ALT{} at
its three weakest tested settings ($\lambda \le 0.005$) and at a cost that grows with the perturbation.

\section{Protocol}
\label{app:protocol}

\subsection{The twelve arms}
\label{app:arms}

All arms share the \RMp{} state format (nonnegative regret vectors per player, zero
initialization, normalized readout, uniform on zero state) and differ only in the
update schedule and prediction structure. \cref{tab:arms} gives the exact rules.

\begin{table}[htbp]
\centering
\caption{The catalog's arms: the twelve main dynamics, the \textsc{alt-prm}$^{+}$
composition of \cref{sec:exp-add}, and the two re-injection dials of
\cref{sec:exp-inj,sec:exp-f3}. $g^i$ denotes player $i$'s instantaneous regret vector
(\cref{sec:setup}); $[\cdot]^{+}$ is the componentwise positive part.}
\label{tab:arms}
\footnotesize
\setlength{\tabcolsep}{4pt}
\begin{tabular}{@{}lp{14.68cm}@{}}
\toprule
\textbf{Arm} & \textbf{Update rule (one round)}\\
\midrule
\SIM{} & both players update from the previous joint profile (Jacobi sweep):
$\vq^i \leftarrow [\vq^i + g^i(\vx_t, \vy_t)]^{+}$.\\
\ALT{} & player 1 updates first; player 2 uses the \emph{new} $\vx_{t+1}$
(\cref{eq:alt}).\\
\PRM{} & two-sided predictive \RMp{}: persistent states accumulate stale increments;
each readout adds the previous round's own increment as prediction
\citep{farina2021faster}.\\
\OSPRM{} & one-sided perfect one-step prediction (\cref{eq:os}): player 2's persistent
state accumulates the stale increment; the readout adds the true fresh increment
$g^2(\vx_{t+1}, \vy_t)$ as a transient term.\\
\GSPROX$(\lambda)$ & \ALT{} schedule with a damped played strategy:
$\vq^1\leftarrow[\vq^1+g^1(\vx_t,\vy_t)]_+$, $\vx^{\mathrm{tgt}}_{t+1}=\sigma(\vq^1)$,
$\vx_{t+1}=(\vx^{\mathrm{tgt}}_{t+1}+\lambda\vx_t)/(1+\lambda)$; then
$\vq^2\leftarrow[\vq^2+g^2(\vx_{t+1},\vy_t)]_+$, $\vy^{\mathrm{tgt}}_{t+1}=\sigma(\vq^2)$,
$\vy_{t+1}=(\vy^{\mathrm{tgt}}_{t+1}+\lambda\vy_t)/(1+\lambda)$. The damping acts on the
played strategy only; the states accumulate exactly as in \ALT{} on the played profiles.
$\lambda \in \{0.25, 1.0\}$; $\lambda = 0$ recovers \ALT{} exactly.\\
\PP{} & proximal-point-style two-sided implicit step approximated by one inner
fixed-point iteration (an extragradient-style step in the spirit of ExRM$^{+}$, \citealp{farina2023regret}; not a reimplementation):
the half-step profile $\tilde\vx_{t+1}=\sigma([\vq^1_t+g^1(\vx_t,\vy_t)]_+)$,
$\tilde\vy_{t+1}=\sigma([\vq^2_t+g^2(\vx_t,\vy_t)]_+)$ is computed and discarded from the
state; the persistent accumulation uses regrets at the opponent's half-step strategy,
$\vq^1_{t+1}=[\vq^1_t+g^1(\vx_t,\tilde\vy_{t+1})]_+$,
$\vq^2_{t+1}=[\vq^2_t+g^2(\tilde\vx_{t+1},\vy_t)]_+$, and the readout
$\vx_{t+1}=\sigma(\vq^1_{t+1})$, $\vy_{t+1}=\sigma(\vq^2_{t+1})$ carries no prediction term.\\
\DELAY-$d$ & player 2's regret uses player 1's strategy from $d$ rounds earlier;
$d \in \{1, 2, 4, 8, 16\}$. By construction $d = 0$ coincides with \ALT{} and
$d = 1$ with \SIM{} (verified bit-exactly).\\
\TSOPT{} & two-sided perfect one-step optimism: both persistent states accumulate
stale increments, $\vq^i_{t+1}=[\vq^i_t+g^i(\vx_t,\vy_t)]_+$; the lookahead strategies
are the plain-\RMp{} readouts of the updated states, $\tilde\vx_{t+1}=\sigma(\vq^1_{t+1})$,
$\tilde\vy_{t+1}=\sigma(\vq^2_{t+1})$; each played strategy adds the regret at the
opponent's lookahead with the own current strategy as baseline,
$\vx_{t+1}=\sigma([\vq^1_{t+1}+g^1(\vx_t,\tilde\vy_{t+1})]_+)$,
$\vy_{t+1}=\sigma([\vq^2_{t+1}+g^2(\tilde\vx_{t+1},\vy_t)]_+)$. Dropping player~1's
prediction term gives $\vx_{t+1}=\tilde\vx_{t+1}$ and recovers \OSPRM{} exactly.\\
alt-prm$^{+}$ & alternating two-sided \PRM: $\vq^1_{t+1}=[\vq^1_t+g^1(\vx_t,\vy_t)]_+$, $\vx_{t+1}=\sigma([\vq^1_{t+1}+g^1(\vx_t,\vy_t)]_+)$; $\vq^2_{t+1}=[\vq^2_t+g^2(\vx_{t+1},\vy_t)]_+$, $\vy_{t+1}=\sigma([\vq^2_{t+1}+g^2(\vx_{t+1},\vy_t)]_+)$ (PCFR$^{+}$ schedule).\\
inj($\mu$) & $\vq^2\leftarrow[\vq^2+(1-\mu)g^2(\vx_t,\vy_t)+\mu\,g^2(\vx_{t+1},\vy_t)]_+$, $\vy_{t+1}=\sigma([\vq^2+(1-\mu)g^2(\vx_{t+1},\vy_t)]_+)$; total fresh weight in the readout is 1; $\mu=0$ is \OSPRM, $\mu=1$ is \ALT{} (bit-exact).\\
mix($\theta$) & $\vq^2\leftarrow[\vq^2+(1-\theta)g^2(\vx_t,\vy_t)+\theta\,g^2(\vx_{t+1},\vy_t)]_+$, $\vy_{t+1}=\sigma(\vq^2)$ (no transient term); $\theta=0$ is \SIM, $\theta=1$ is \ALT{} (bit-exact).\\
\bottomrule
\end{tabular}
\end{table}

\subsection{Game suite}

The 47-game matrix suite comprises: the 12-game pilot family (random dense/sparse
payoffs of sizes up to $5\times 5$, plus structured instances including a biased
matching-pennies); a 20-game size family (random games of increasing dimension); and a
15-game conditioning family ($U\,\mathrm{diag}(s)\,V^{\top}$ with orthogonal factors
and log-spaced spectra at $\kappa \in \{10^2, 10^4, 10^6\}$). The conditioning sweep of
\cref{sec:exp-f4} extends this construction to $\kappa \in \{10^1, \dots, 10^8\}$ with
5 fresh instances per band. Extensive-form witnesses: Kuhn poker (55 nodes / 12
information sets) \citep{kuhn2016simplified}, Leduc poker (9451 / 936)
\citep{southey2005bayes}, and a liar's-dice game with 4-face dice (8177 / 1024)
\citep{neller2011approximating}; the 5- and 6-face versions ($51{,}176$ / $294{,}877$
nodes) are the scale test of \cref{sec:exp-big}. Extensive-form
dynamics use the standard CFR decomposition with \RMp{} at each information set and
the corresponding schedule (alternating, simultaneous, predictive) applied treewide.
On extensive-form games the identities of \cref{prop:identity} are applied in sequence
form: $\vx_t,\vy_t$ are the realization plans of the behavioral profiles, $u$ is bilinear
in them, the averages are reach-weighted (realization-plan) averages, and
$R^1,R^2_{\mathrm{stale}},R^2_{\mathrm{fresh}}$ are the global sequence-form regrets of
the averaged plans, not sums of local counterfactual regrets. \cref{prop:closed} is a
statement about a single simplex state and is not invoked on trees; the sign $\Delta\ge0$
is measured there (all ten tree runs of \cref{tab:big5,tab:big6}), not proved.

\subsection{Lockstep protocol and determinism}

Every arm is run in lockstep with an \ALT{} reference on the same game: a single
program advances both dynamics and records $D_t$ and gap trajectories, eliminating any
cross-run nondeterminism. Lockstep references were additionally validated against
independently-run \ALT{} arms (maximal relative final-gap deviation: $0$). All
computations are pure CPU, single-process, in double precision unless stated (the
four-precision recomputation of \cref{tab:prec} also uses extended precision and
multi-precision arithmetic at 30 and 60 digits), with a fixed seed (game generation
only; the dynamics are deterministic given the game); total compute is $\approx 2.5$ CPU-hours for the main
suite on a single x86-64 server core. The extension runs add roughly $1.2$
CPU-hours for the four-precision $60$-point recomputation ($520$\,s on eight
workers), $0.4$ and $3.0$ CPU-hours for the 5- and 6-face liar's-dice runs
($293$\,s and $2193$\,s on five workers), and about a second for the epoch audit. Matrix-suite horizons are
$T = 2\times 10^4$ ($10^6$ for deep verification), extensive-form horizons $T = 10^4$,
witness-scan certificates $T = 2\times 10^5$.

\emph{Trajectory statistics.} \emph{Drift} is the total variation of the played
profile, $\sum_t \|\vx_{t+1}-\vx_t\|_1 + \|\vy_{t+1}-\vy_t\|_1$; \emph{oscillation}
is the same sum applied to successive differences,
$\sum_t \|(\vx_{t+1}-\vx_t)-(\vx_t-\vx_{t-1})\|_1$ plus the analogous $y$-term; an
arm \emph{hugs} \ALT{} when the profile distance $D_t$ has a decaying trend (tail mean
below half the head mean) and its gap ratio to \ALT{} lies in the reported band.

\emph{Last-iterate curves.}\label{app:lastiter-metric}
The plotted quantity is the first-coordinate deviation $|x_t-x^{*}|+|y_t-y^{*}|$ for
the two-action witness, and the full profile $\ell_1$ distance
$\|\vx_t-\vx^{*}\|_1+\|\vy_t-\vy^{*}\|_1$ is twice this value. Horizons in
\cref{tab:prec,tab:lastiter} are in iterations.

The identities of \cref{prop:closed,prop:order} hold on the suite arms whose player~1
is plain \RMp, namely \ALT, \SIM, \OSPRM, the \textsc{inj} and \textsc{mix} dials of
\cref{sec:exp-inj,sec:exp-f3}, and the \DELAY family.

\subsection{Sanity checks}

The update rules of \cref{eq:alt,eq:os} were verified line-by-line against the
implementation, with the first ten iterations replayed in exact rational arithmetic to
$10^{-12}$. Pilot-family results were reproduced bit-exactly against an earlier
independent implementation. In the course of that reproduction we corrected one overstated summary
in our own pilot notes: two-sided \PRM{} is faster than \ALT{} on 9/12 pilot games at
$T = 10^4$ (11/12 at $T = 2\times 10^4$), not 12/12 as first recorded; the core
separation \ALT{} vs \OSPRM{} is 12/12 at both horizons on the pilot family (median
$25.3\times$). On the
liar's-dice witness, the uniform-strategy exploitability is $0.804$ and the \ALT{}
curve decreases monotonically to $5.45\times 10^{-8}$ over $T = 10^4$.

\section{Results}
\label{app:tables}

\subsection{F1 boundary statistics (47-game matrix suite)}

\begin{table}[H]
\centering
\caption{Suite medians for the F1 arm and the damping family, lockstep vs \ALT{}
($T = 2\times 10^4$). ``Hug'' counts use the baseline band $[\nicefrac13, 3]$ combined
with trajectory closeness; they are interpretable summaries, not the primary
criterion (\cref{sec:setup}).}
\label{tab:f1full}
\scriptsize
\begin{tabular}{@{}lccccc@{}}
\toprule
\textbf{Arm} & gap/\ALT{} (med) & drift/\ALT{} & oscillation/\ALT{} & final $L_1$ (med) & hug count\\
\midrule
\OSPRM{} & 32.9 & 65.9 & 263 & 0.577 & 8/47\\
\GSPROX$(0.25)$ & 19.8 & -- & -- & -- & 7/47\\
\GSPROX$(1.0)$ & 56.5 & -- & -- & -- & 0/47\\
\bottomrule
\end{tabular}
\end{table}

Per-arm baseline-band hug counts across all arms:
\SIM{} 5, \PRM{} 11, \OSPRM{} 8, \GSPROX$(0.25)$ 7, \GSPROX$(1.0)$ 0, \PP{} 13,
\DELAY-1/2/4/8 each 5, \TSOPT{} 14 (of 47; concentrated in the conditioning family).
No arm reaches the $90\%$ reproduction criterion in any band. In the widest band
$[\nicefrac15, 5]$, \GSPROX$(1.0)$ rises from 0 to 6/47. The six games are
\texttt{ext5x5\_2}, \texttt{cond1e2\_3}, \texttt{cond1e2\_4}, \texttt{cond1e4\_2},
\texttt{cond1e6\_0}, and \texttt{cond1e6\_3}, disclosed per pre-registration.

\subsection{Staleness curve and robustness}

\begin{table}[H]
\centering
\caption{Median final gap / \ALT{} over the 47-game suite along the staleness axis.
Ratio rows are medians over games of the per-game ratio. Strict per-game monotonicity
violations: 4 of 47 games
(largest: $43.9\%$ decrease at $d\,{=}\,1\!\to\!2$ on one random $5\times5$ game;
three conditioning-family games show $1$--$17\%$ dips at $d\,{=}\,0\!\to\!1$).}
\label{tab:delayfull}
\footnotesize
\begin{tabular}{@{}lcccccc@{}}
\toprule
$d$ & 0 (\ALT) & 1 (\SIM) & 2 & 4 & 8 & 16\\
\midrule
median gap/\ALT{} & 1 & 41.33 & 78.86 & 166.5 & 471.0 & 1407.7\\
per-doubling ratio & -- & (one step) & 1.91 & 2.12 & 2.27 & 2.42\\
\bottomrule
\end{tabular}
\end{table}

Staleness robustness (50 independent random games, seed 43): median ratio $39.4$,
bootstrap $95\%$ CI $[29.7, 51.0]$ ($10^4$ resamples); distribution quantiles
(min/25/50/75/max): $0.79 / 24.5 / 39.4 / 69.1 / 1778$; $d{=}2$ ratio median $79.0$,
CI $[64.3, 98.6]$.

\subsection{Conditioning sweep (full medians)}

\begin{table}[H]
\centering
\caption{Median final gap / \ALT{} per conditioning band (5 instances each,
$T = 2\times 10^4$). \SIM{} row = the \ALT{} advantage (values $>1$ mean \ALT{}
faster); optimism rows $<1$ mean the arm is faster than \ALT{}. Spearman rank
correlation of the ratio with $\kappa$ in the last column.}
\label{tab:kappafull}
\footnotesize
\begin{tabular}{@{}lccccccccc@{}}
\toprule
\textbf{Arm} & $10^1$ & $10^2$ & $10^3$ & $10^4$ & $10^5$ & $10^6$ & $10^7$ & $10^8$ & Sp.\\
\midrule
\SIM{} & 37.0 & 13.5 & 149.1 & 40.7 & 87.9 & 21.2 & 33.3 & 71.5 & --\\
\PRM{} & 0.041 & 0.018 & 0.180 & 0.191 & 0.779 & 1.594 & 3.899 & 2.664 & 0.95\\
\PP{} & 0.001 & 0.008 & 0.078 & 0.054 & 0.364 & 1.212 & 1.053 & 0.555 & 0.88\\
\TSOPT{} & 0.046 & 0.006 & 0.324 & 0.048 & 0.308 & 1.000 & 0.699 & 0.723 & 0.83\\
\bottomrule
\end{tabular}
\end{table}

Closure thresholds (first band where gap/\ALT{} $\ge \nicefrac12$):
$\kappa^{*}(\text{\PRM{}}) = 10^5$, $\kappa^{*}(\text{\PP{}}) = \kappa^{*}(\text{\TSOPT{}}) = 10^6$. Under the
vs-\SIM{} operationalization (optimism/\SIM{} advantage falling below $2\times$), no
arm triggers within $\kappa \le 10^8$: the advantage shrinks from $10^3$--$10^5\times$
to $5.8$--$107\times$. Data quality: all final gaps $\ge 10^{-9}$ (no clamping-floor
artifacts); lockstep reference deviation $0$.

\subsection{Extensive-form witnesses (full ranking)}

\begin{table}[H]
\centering
\caption{Final exploitability and ratio to \ALT{} on the three extensive-form
witnesses ($T = 10^4$). Kuhn and Leduc columns give the exploitability ratio to
\ALT{} (best-response sum divided by two; \cref{sec:setup}, \emph{Conventions});
liar's dice gives both the final exploitability and the ratio. Bold: within the
$3\times$ reporting band of \ALT{}.}
\label{tab:efgfull}
\scriptsize
\setlength{\tabcolsep}{4pt}
\begin{tabular}{@{}lcccc@{}}
\toprule
\textbf{Arm} & Kuhn & Leduc & Liar's dice & Liar's ratio\\
\midrule
\ALT{} & 1 & 1 & $5.450\times 10^{-8}$ & 1\\
\PP{} & 0.0019 & 19.1 & $5.817\times 10^{-8}$ & \textbf{1.07}\\
\TSOPT{} & \textbf{2.04} & 70.4 & $7.450\times 10^{-8}$ & \textbf{1.37}\\
\GSPROX$(0.25)$ & 40.4 & 87.7 & $8.538\times 10^{-8}$ & \textbf{1.57}\\
\OSPRM{} & 74.9 & $1.54\times 10^{3}$ & $1.025\times 10^{-7}$ & \textbf{1.88}\\
\PRM{} & 0.070 & 32.1 & $2.059\times 10^{-7}$ & 3.78\\
\SIM{} (= \DELAY-1) & 118.6 & 320.8 & $4.358\times 10^{-7}$ & 7.99\\
\GSPROX$(1.0)$ & 142.7 & 388.7 & $9.688\times 10^{-7}$ & 17.8\\
\DELAY-2 & 190.9 & 656.1 & $2.100\times 10^{-6}$ & 38.5\\
\DELAY-4 & 423.3 & $1.69\times 10^{3}$ & $2.404\times 10^{-5}$ & 441\\
\DELAY-8 & 903.5 & $4.78\times 10^{3}$ & $2.412\times 10^{-4}$ & $4.43\times 10^{3}$\\
\bottomrule
\end{tabular}
\end{table}

Trajectory level on liar's dice: final $L_1$ distance to \ALT{} is $10.1$--$23.3$
across all arms (the four gap-band arms: \PP{} 23.3, \TSOPT{} 22.7,
\GSPROX$(0.25)$ 16.3, \OSPRM{} 10.1), all flat in $t$: no arm hugs \ALT{} at the
trajectory level even on the boundary game. On Leduc, the trajectory-closeness
(``decaying $L_1$'') summary has weak discriminative power (all arms decay while final
values remain large), so extensive-form conclusions rest on gap ratios; disclosed here.

\subsection{Witness-family certificates (theorem support)}

\begin{table}[H]
\centering
\caption{Certificate summary over the 60-point grid $a \in [1.05, 4]$
($T = 2\times 10^5$, tail-half statistics), supporting \cref{ass:c1} and the
evidence against branch (b). The full per-point table is in the
certificate file.}
\label{tab:certs}
\scriptsize
\begin{tabularx}{\linewidth}{@{}l X@{}}
\toprule
\textbf{Certificate} & \textbf{Value (range over 60 points)}\\
\midrule
C1 interior band $\eta(a)$ & $0.265$--$0.453$ (median $0.339$); \ALT{} truncation-free tail: 60/60\\
predicted constant $c = 2\eta(a)$ & $0.530$--$0.905$\\
measured $\min_t D_t$ at vertex times & $0.530$--$0.906$; matches $2\eta$ within $9.0\times 10^{-4}$ rel.\\
exact-vertex frequency of $\vy^{\mathrm{OS}}$ & $0.333$--$0.950$ (median $0.518$); positive: 60/60\\
\OSPRM{} persistent state $\bar\varepsilon$ & $4.9\times 10^{-6}$--$8.5\times 10^{-3}$\\
overall certificate (all conditions) & 60/60\\
\bottomrule
\end{tabularx}
\end{table}

Representative point $a = 2$: $\eta = 0.37232$, $2\eta = 0.74464$, measured
$\min_t D_t$ at vertex times $= 0.74484$, vertex frequency $0.657$,
$\bar\varepsilon = 1.9\times 10^{-3}$. The un-injected side of \OSPRM{} sticks near
$x^{*}$: measured $|\delta^2(x_t)|$ tail p99 $\le 0.026$ across all 60 grid points.

\subsection{Four-precision recomputation of the certificates}
\label{app:prec}

The certificate statistics of \cref{tab:certs} recomputed in four arithmetic
precisions, run in lockstep over the same 60-point grid.

\begin{table}[H]
\centering\scriptsize
\setlength{\tabcolsep}{4pt}
\caption{Four-precision recomputation of the witness certificates (60 grid points,
$T=2\times10^5$, tail half; four arithmetic precisions, double, extended, and 30- and
60-digit multi-precision, run in lockstep). Statistics agree across precisions although
individual \OSPRM{} trajectories do not; the quoted cross-precision differences are
relative and compare the double run with the 60-digit run.}\label{tab:prec}
\begin{tabular*}{\linewidth}{@{}@{\extracolsep{\fill}}p{4.6cm}cccc@{}}
\toprule
 & double & extended & 30-digit & 60-digit\\
\midrule
$\eta(a)$ range & \multicolumn{4}{c}{$[0.265, 0.453]$ (identical to $1.6\times10^{-12}$ rel.)}\\
exact-vertex frequency range & \multicolumn{4}{c}{$[0.333, 0.950]$ (max cross-precision difference $7.4\times10^{-3}$)}\\
\quad jump / stay (median) & \multicolumn{4}{c}{$0.300$ / $0.203$}\\
$\bar\varepsilon$ max & \multicolumn{4}{c}{$8.5\times10^{-3}$ (max cross-precision difference $1.6\%$)}\\
$\min_t D_t$ at vertex times $-\,2\eta$, rel.\ ($10^{-5}$) & $[5.3, 90]$ & $[8.9, 113]$ & $[4.0, 92]$ & $[3.3, 81]$\\
\ALT{} truncation-free tail & \multicolumn{4}{c}{60/60}\\
\OSPRM{} trajectory diverged from double ($|\Delta y|>0.1$) & -- & 53/60 (median $t=229$) & 53/60 & 53/60\\
\ALT{} max $|\Delta y|$ vs double & -- & \multicolumn{3}{c}{$4.8\times10^{-11}$}\\
jump-time margin ratio $\varepsilon_t/(\theta_t\Delta_t+\varepsilon_t)$: median / max & \multicolumn{4}{c}{$0.21$ / $1.0$}\\
\bottomrule
\end{tabular*}
\end{table}

% ---------------------------------------------------------------------------
% D_tables_ext.tex — generated by make_ext_tables.py and since hand-edited
% (R4/R6 review rounds: tab:anatomy shifted-gap column, 7-column tab:epoch,
% appended tab:dials). Do NOT rerun make_ext_tables.py; it would drop the
% edits. New content goes through review-stage scripts.
% ---------------------------------------------------------------------------
\subsection{Exact decomposition}

The exact decomposition of the duality gap (\cref{eq:identity}) evaluated on
four witness points, per unit weight $W$, with the stale and fresh readings of
the opponent's strategy separated.

\begin{table}[H]
\centering
\caption{Exact decomposition on the four witness games at $T = 2\times 10^{4}$, per unit weight $W$.}\label{tab:anatomy}
\footnotesize
\setlength{\tabcolsep}{3.5pt}
\begin{tabular}{@{}llrrrrrrr@{}}
\toprule
$a$ & arm & $R^1/W$ & $R^2_{\mathrm{stale}}/W$ & $R^2_{\mathrm{fresh}}/W$ & $\Delta/W$ & $\mathrm{gap}(\bar\vx,\bar\vy)/2$ & $\mathrm{gap}(\bar\vx^{+},\bar\vy)/2$ & vertex\\
\midrule
1.2 & \SIM & $5.88\times 10^{-3}$ & $7.93\times 10^{-3}$ & $1.50\times 10^{-2}$ & $6.99\times 10^{-3}$ & $6.90\times 10^{-3}$ & $6.93\times 10^{-3}$ & $0.054$\\
1.2 & \ALT & $4.79\times 10^{-4}$ & $-2.19\times 10^{-4}$ & $4.51\times 10^{-4}$ & $6.78\times 10^{-4}$ & $1.30\times 10^{-4}$ & $1.26\times 10^{-4}$ & $0.000$\\
1.2 & \OSPRM & $4.98\times 10^{-3}$ & $-3.52\times 10^{-3}$ & $5.93\times 10^{-3}$ & $9.45\times 10^{-3}$ & $7.32\times 10^{-4}$ & $7.32\times 10^{-4}$ & $0.821$\\
1.2 & \PRM & $6.06\times 10^{-5}$ & $6.66\times 10^{-5}$ & $3.89\times 10^{-1}$ & $3.89\times 10^{-1}$ & $6.36\times 10^{-5}$ & $3.85\times 10^{-5}$ & $0.429$\\
1.2 & \textsc{alt-prm}$^{+}$ & $2.49\times 10^{-8}$ & $-1.04\times 10^{-8}$ & $2.14\times 10^{-8}$ & $3.50\times 10^{-8}$ & $7.25\times 10^{-9}$ & $5.63\times 10^{-9}$ & $0.000$\\
2 & \SIM & $5.52\times 10^{-3}$ & $6.53\times 10^{-3}$ & $1.30\times 10^{-2}$ & $6.49\times 10^{-3}$ & $6.03\times 10^{-3}$ & $6.02\times 10^{-3}$ & $0.060$\\
2 & \ALT & $3.93\times 10^{-4}$ & $-1.88\times 10^{-4}$ & $3.99\times 10^{-4}$ & $5.92\times 10^{-4}$ & $1.03\times 10^{-4}$ & $1.00\times 10^{-4}$ & $0.000$\\
2 & \OSPRM & $6.10\times 10^{-3}$ & $-3.24\times 10^{-3}$ & $7.24\times 10^{-3}$ & $1.05\times 10^{-2}$ & $1.43\times 10^{-3}$ & $1.43\times 10^{-3}$ & $0.645$\\
2 & \PRM & $1.41\times 10^{-4}$ & $3.31\times 10^{-5}$ & $3.04\times 10^{-1}$ & $3.04\times 10^{-1}$ & $8.73\times 10^{-5}$ & $1.23\times 10^{-4}$ & $0.124$\\
2 & \textsc{alt-prm}$^{+}$ & $3.50\times 10^{-8}$ & $-1.89\times 10^{-8}$ & $2.79\times 10^{-8}$ & $5.25\times 10^{-8}$ & $8.03\times 10^{-9}$ & $5.18\times 10^{-9}$ & $0.000$\\
3 & \SIM & $3.36\times 10^{-3}$ & $5.25\times 10^{-3}$ & $1.19\times 10^{-2}$ & $6.59\times 10^{-3}$ & $4.31\times 10^{-3}$ & $4.33\times 10^{-3}$ & $0.064$\\
3 & \ALT & $3.60\times 10^{-4}$ & $-2.95\times 10^{-4}$ & $3.18\times 10^{-4}$ & $6.11\times 10^{-4}$ & $3.25\times 10^{-5}$ & $3.31\times 10^{-5}$ & $0.000$\\
3 & \OSPRM & $6.71\times 10^{-3}$ & $-1.82\times 10^{-3}$ & $9.46\times 10^{-3}$ & $1.13\times 10^{-2}$ & $2.45\times 10^{-3}$ & $2.45\times 10^{-3}$ & $0.545$\\
3 & \PRM & $-4.12\times 10^{-7}$ & $1.97\times 10^{-6}$ & $1.28\times 10^{-4}$ & $1.26\times 10^{-4}$ & $7.79\times 10^{-7}$ & $7.69\times 10^{-7}$ & $0.001$\\
3 & \textsc{alt-prm}$^{+}$ & $5.89\times 10^{-8}$ & $-3.14\times 10^{-8}$ & $4.39\times 10^{-8}$ & $8.34\times 10^{-8}$ & $1.37\times 10^{-8}$ & $9.67\times 10^{-9}$ & $0.000$\\
4 & \SIM & $2.19\times 10^{-3}$ & $3.17\times 10^{-3}$ & $1.02\times 10^{-2}$ & $6.95\times 10^{-3}$ & $2.68\times 10^{-3}$ & $2.70\times 10^{-3}$ & $0.072$\\
4 & \ALT & $4.60\times 10^{-4}$ & $-2.49\times 10^{-4}$ & $4.58\times 10^{-4}$ & $7.11\times 10^{-4}$ & $1.05\times 10^{-4}$ & $1.03\times 10^{-4}$ & $0.000$\\
4 & \OSPRM & $7.78\times 10^{-3}$ & $-1.34\times 10^{-3}$ & $1.14\times 10^{-2}$ & $1.27\times 10^{-2}$ & $3.22\times 10^{-3}$ & $3.22\times 10^{-3}$ & $0.507$\\
4 & \PRM & $9.84\times 10^{-6}$ & $4.63\times 10^{-5}$ & $1.16\times 10^{-1}$ & $1.16\times 10^{-1}$ & $2.81\times 10^{-5}$ & $5.36\times 10^{-5}$ & $0.001$\\
4 & \textsc{alt-prm}$^{+}$ & $9.98\times 10^{-8}$ & $-5.49\times 10^{-8}$ & $7.63\times 10^{-8}$ & $1.42\times 10^{-7}$ & $2.24\times 10^{-8}$ & $1.72\times 10^{-8}$ & $0.000$\\
\bottomrule
\end{tabular}
\end{table}

% 2026-10-14 arXiv 版: tab:inj 与 tab:lastiter 已搬入正文 sections/6_experiments.tex。

% No \clearpage here (removed 2026-10-08): it stranded the table stream on
% half-empty pages, and the overfull it was added to fix does not return
% without it.

\subsection{Epoch audit}

Epoch-level statistics on the witness family, $T = 2\times10^{5}$, tail half,
supporting \cref{prop:epoch,lem:flip}.

\begin{table}[H]
\centering
\caption{Epoch audit on the witness family, $T = 2\times10^{5}$, tail half.
\cref{lem:flip} is evaluated only on length-one epochs
(all $82237$/$50105$/$39564$/$27582$ of them); its prediction agrees with the measured
readout at every one. Vertex times inside no epoch: $0$/$0$/$1$/$0$.}
\label{tab:epoch}
\footnotesize
\setlength{\tabcolsep}{4pt}
\begin{tabular}{@{}lrrrrrrr@{}}
\toprule
$a$ & epochs & \shortstack[r]{len-1 share\\(\%)} & \shortstack[r]{\eqref{eq:washout-criterion},\\len $\ge 2$ (\%)} & \shortstack[r]{zero entry,\\len-1 (\%)} & \shortstack[r]{vertex\\epochs (\%)} & \shortstack[r]{len-1 exits among\\vertex times (\%)} & \shortstack[r]{zero disfav.\ at\\vertex (\%)}\\
\midrule
1.2 & 88158 & 93.3 & 100.0 & 78.4 & 86.6 & 85.6 & 92.8\\
2.0 & 66807 & 75.0 & 98.7 & 32.8 & 73.6 & 49.8 & 74.9\\
3.0 & 60284 & 65.6 & 66.8 & 14.7 & 57.4 & 29.4 & 60.0\\
4.0 & 51231 & 53.8 & 42.6 & 25.1 & 61.3 & 26.0 & 69.2\\
\bottomrule
\end{tabular}
\end{table}

\subsection{Perturbation sensitivity and the three dials}

Sensitivity of the four primary arms to a relative perturbation of $10^{-13}$
under two injection modes (payoffs scaled, initial regrets jittered), and the
verdicts of the three dials (mixing, damping, injection) under the primary and
the $3\times$ criteria.

\begin{table}[H]
\centering
\caption{Sensitivity to a relative perturbation of size $10^{-13}$, by mode (payoffs scaled vs.\ initial regrets jittered), $T = 2\times 10^{4}$, 47 games. `macro' counts games whose $\ell_1$ divergence exceeds $0.1$; $t_{\mathrm{first}}$ is the median breach round.}
\label{tab:sens}
\small
\setlength{\tabcolsep}{3.5pt}
\begin{tabular}{@{}llcccc@{}}
\toprule
mode & arm & $D_{\max}$ median & $D_{\max}$ max & macro & $t_{\mathrm{first}}$\\
\midrule
init & \ALT & $2.6\times 10^{-11}$ & $1.0\times 10^{-6}$ & 0/47 & --\\
init & \SIM & $5.1\times 10^{-12}$ & $4.2\times 10^{-1}$ & 1/47 & 881\\
init & \OSPRM & $2.0\times 10^{0}$ & $2.2\times 10^{0}$ & 31/47 & 97\\
init & \PRM & $2.4\times 10^{-12}$ & $2.6\times 10^{0}$ & 5/47 & 743\\
scale & \ALT & $2.4\times 10^{-12}$ & $2.8\times 10^{-11}$ & 0/47 & --\\
scale & \SIM & $7.3\times 10^{-13}$ & $4.1\times 10^{-1}$ & 1/47 & 1086\\
scale & \OSPRM & $2.0\times 10^{0}$ & $2.2\times 10^{0}$ & 33/47 & 182\\
scale & \PRM & $7.0\times 10^{-15}$ & $2.6\times 10^{0}$ & 5/47 & 806\\
\bottomrule
\end{tabular}
\end{table}

% BEGIN tab:dials (generated by review-stage/r6_dial_table.py; do not edit by hand)
\begin{table}[H]
\centering
\caption{Verdicts of the three dials under the primary and the $3\times$ criteria.}
\label{tab:dials}
\footnotesize
\setlength{\tabcolsep}{4pt}
\begin{tabular}{@{}llcclc@{}}
\toprule
dial & setting & median $\rho_g$ [IQR] & within $3\times$ & median criterion & $D$-tail\\
\midrule
mixing & $0$ & $+0.00$ $[0.00, 0.00]$ & 47/47 & yes & $0.000$\\
 & $0.01$ & $+0.42$ $[0.17, 0.69]$ & 26/47 & yes & $0.138$\\
 & $0.02$ & $+0.60$ $[0.26, 0.79]$ & 20/47 & no & $0.138$\\
 & $0.05$ & $+0.80$ $[0.46, 1.07]$ & 13/47 & no & $0.144$\\
 & $0.1$ & $+0.87$ $[0.52, 1.26]$ & 10/47 & no & $0.114$\\
 & $0.25$ & $+1.14$ $[0.79, 1.39]$ & 7/47 & no & $0.133$\\
 & $0.5$ & $+1.27$ $[1.07, 1.52]$ & 7/47 & no & $0.174$\\
 & $0.75$ & $+1.48$ $[1.26, 1.70]$ & 6/47 & no & $0.238$\\
 & $0.9$ & $+1.62$ $[1.31, 1.75]$ & 5/47 & no & $0.277$\\
 & $1$ & $+1.62$ $[1.39, 1.81]$ & 5/47 & no & $0.279$\\
\addlinespace[2pt]
damping & $0.001$ & $+0.22$ $[0.00, 0.38]$ & 39/47 & yes & $0.103$\\
 & $0.002$ & $+0.29$ $[0.01, 0.57]$ & 31/47 & yes & $0.106$\\
 & $0.005$ & $+0.47$ $[0.14, 0.65]$ & 24/47 & yes & $0.134$\\
 & $0.01$ & $+0.61$ $[0.30, 0.82]$ & 15/47 & no & $0.142$\\
 & $0.02$ & $+0.75$ $[0.44, 0.99]$ & 14/47 & no & $0.144$\\
 & $0.05$ & $+0.91$ $[0.55, 1.23]$ & 10/47 & no & $0.131$\\
 & $0.1$ & $+1.00$ $[0.75, 1.37]$ & 8/47 & no & $0.166$\\
 & $0.25$ & $+1.30$ $[1.03, 1.60]$ & 7/47 & no & $0.162$\\
 & $0.5$ & $+1.52$ $[1.25, 1.73]$ & 6/47 & no & $0.229$\\
 & $1$ & $+1.75$ $[1.52, 2.09]$ & 0/47 & no & $0.272$\\
 & $2$ & $+1.92$ $[1.71, 2.45]$ & 0/47 & no & $0.337$\\
 & $4$ & $+2.18$ $[1.88, 2.74]$ & 0/47 & no & $0.469$\\
\addlinespace[2pt]
injection & $0$ & $+1.52$ $[0.96, 2.08]$ & 8/47 & no & $0.686$\\
 & $0.1$ & $+1.61$ $[0.98, 2.08]$ & 9/47 & no & $0.669$\\
 & $0.25$ & $+1.51$ $[1.09, 2.10]$ & 9/47 & no & $0.686$\\
 & $0.3$ & $+1.46$ $[1.06, 2.04]$ & 9/47 & no & $0.650$\\
 & $0.35$ & $+1.32$ $[0.51, 1.81]$ & 12/47 & no & $0.588$\\
 & $0.4$ & $+1.20$ $[0.14, 1.67]$ & 14/47 & no & $0.508$\\
 & $0.45$ & $+0.58$ $[-0.02, 1.05]$ & 21/47 & no & $0.320$\\
 & $0.5$ & $+0.22$ $[-0.08, 0.57]$ & 27/47 & yes & $0.072$\\
 & $0.75$ & $-0.03$ $[-0.11, 0.08]$ & 43/47 & yes & $0.054$\\
 & $0.9$ & $-0.02$ $[-0.13, 0.02]$ & 46/47 & yes & $0.044$\\
 & $1$ & $+0.00$ $[0.00, 0.00]$ & 47/47 & yes & $0.000$\\
\bottomrule
\end{tabular}
\end{table}
% END tab:dials

% Auto-generated by gen_pergame_table.py — do not edit by hand.
% Source: alt_opt/results.json matrix.per_game; rho = log10(gap_arm / gap_ALT).
\subsection{Per-game log-ratios}

Per-game values behind the suite medians, one row per game.

\begin{table}[H]
\centering
\scriptsize
\setlength{\tabcolsep}{2.0pt}
\caption{Per-game log-ratios $\rho_g = \log_{10}(\mathrm{gap}_{\mathrm{arm}}/\mathrm{gap}_{\ALTmath})$ on the 47-game suite ($T = 2\times10^4$, seed 42), one row per game; positive = slower than \ALT. Columns as in \cref{tab:rho}; GS.25/GS1 = \GSPROX$(0.25/1.0)$, D1--D8 = \DELAY-$d$.}
\label{tab:pergame}
\begin{tabular}{@{}l*{11}{r}@{}}
\toprule
\textbf{Game} & \SIM & \PRM & \OSPRM & GS.25 & GS1 & \PP & D1 & D2 & D4 & D8 & \TSOPT\\
\midrule
\texttt{rand3x3\_0} & +2.13 & -2.49 & +1.81 & +2.09 & +2.40 & -2.94 & +2.13 & +2.42 & +2.65 & +3.01 & -2.96\\
\texttt{rand3x3\_1} & +1.47 & -2.59 & +1.31 & +1.57 & +1.65 & -3.31 & +1.47 & +1.96 & +2.17 & +2.51 & -3.22\\
\texttt{rand3x3\_2} & +1.62 & -0.30 & +1.13 & +1.40 & +1.77 & -3.31 & +1.62 & +1.88 & +2.14 & +2.46 & -3.44\\
\texttt{rand3x3\_3} & +2.45 & -2.18 & +2.02 & +2.31 & +2.62 & -3.19 & +2.45 & +2.68 & +2.92 & +3.21 & -2.57\\
\texttt{rand3x3\_4} & +1.70 & -0.70 & +1.44 & +1.30 & +1.86 & -2.85 & +1.70 & +2.00 & +2.34 & +2.73 & +0.27\\
\texttt{rand5x5\_0} & +1.64 & -2.34 & +1.33 & +1.20 & +1.73 & -3.53 & +1.64 & +1.85 & +2.06 & +2.49 & -2.24\\
\texttt{rand5x5\_1} & +1.47 & -2.70 & +1.52 & +1.34 & +1.56 & -3.57 & +1.47 & +1.22 & +1.76 & +2.17 & -3.32\\
\texttt{rand5x5\_2} & +1.73 & -0.21 & +0.93 & +1.66 & +2.04 & -3.61 & +1.73 & +1.91 & +2.33 & +2.67 & -0.20\\
\texttt{rand5x5\_3} & +1.42 & +0.71 & +2.53 & +1.03 & +1.51 & -0.25 & +1.42 & +1.56 & +1.95 & +2.45 & +0.98\\
\texttt{rand5x5\_4} & +1.31 & -3.15 & +1.44 & +1.05 & +1.52 & -3.63 & +1.31 & +1.59 & +1.82 & +2.25 & -3.03\\
\texttt{biased\_mp} & +1.77 & -0.09 & +1.14 & +1.52 & +1.81 & -4.50 & +1.77 & +1.92 & +2.01 & +2.49 & -3.76\\
\texttt{rps\_variant} & +1.50 & -0.34 & +1.37 & +1.75 & +1.79 & -4.51 & +1.50 & +1.92 & +2.12 & +2.47 & -0.26\\
\midrule
\texttt{ext3x3\_0} & +1.74 & -0.18 & +0.80 & +1.76 & +1.64 & -4.08 & +1.74 & +1.84 & +2.31 & +2.46 & -3.78\\
\texttt{ext3x3\_1} & +1.80 & -3.02 & +0.34 & +1.59 & +1.99 & -3.90 & +1.80 & +2.01 & +2.20 & +2.45 & -4.04\\
\texttt{ext3x3\_2} & +1.69 & -2.06 & +1.42 & +1.40 & +1.81 & -2.99 & +1.69 & +2.00 & +2.34 & +2.70 & +1.12\\
\texttt{ext3x3\_3} & +1.66 & -0.35 & +0.61 & +1.57 & +1.71 & -4.39 & +1.66 & +1.81 & +2.10 & +2.28 & -0.58\\
\texttt{ext3x3\_4} & +1.57 & -2.69 & +1.01 & +1.60 & +1.83 & -4.37 & +1.57 & +1.64 & +2.18 & +2.40 & -3.30\\
\texttt{ext5x5\_0} & +1.04 & -2.16 & +1.96 & +0.86 & +1.34 & -3.01 & +1.04 & +1.77 & +2.42 & +2.83 & -2.70\\
\texttt{ext5x5\_1} & +1.97 & -1.91 & +2.08 & +1.20 & +2.33 & -2.53 & +1.97 & +2.69 & +2.93 & +3.34 & +0.78\\
\texttt{ext5x5\_2} & +0.07 & +0.07 & +0.00 & +0.11 & +0.49 & -0.18 & +0.07 & +0.18 & +0.18 & +0.18 & -0.04\\
\texttt{ext5x5\_3} & +1.81 & -1.67 & +2.14 & +1.08 & +2.16 & -2.45 & +1.81 & +2.67 & +3.01 & +3.30 & -2.15\\
\texttt{ext5x5\_4} & +2.08 & -2.51 & +1.66 & +1.03 & +2.27 & -3.18 & +2.08 & +2.40 & +2.70 & +3.08 & -2.88\\
\texttt{ext10x10\_0} & +1.33 & -2.10 & +2.11 & +1.30 & +1.52 & -2.59 & +1.33 & +1.57 & +2.49 & +2.99 & -2.44\\
\texttt{ext10x10\_1} & +1.37 & -2.23 & +1.53 & +1.19 & +1.49 & -2.59 & +1.37 & +1.71 & +2.22 & +2.58 & +0.02\\
\texttt{ext10x10\_2} & +2.83 & -0.45 & +2.90 & +1.99 & +2.93 & -0.85 & +2.83 & +3.10 & +3.40 & +3.81 & -0.64\\
\texttt{ext10x10\_3} & +1.68 & -2.31 & +1.46 & +1.60 & +1.84 & -3.15 & +1.68 & +1.90 & +2.19 & +2.51 & -2.72\\
\texttt{ext10x10\_4} & +1.75 & -1.50 & +2.16 & +1.55 & +1.86 & -1.54 & +1.75 & +2.29 & +2.87 & +3.18 & -1.96\\
\texttt{ext20x20\_0} & +2.90 & +1.25 & +3.80 & +2.60 & +3.12 & +0.02 & +2.90 & +3.50 & +3.98 & +4.38 & +1.09\\
\texttt{ext20x20\_1} & +2.02 & -2.11 & +2.08 & +1.70 & +2.16 & -2.39 & +2.02 & +2.36 & +2.71 & +3.04 & -2.24\\
\texttt{ext20x20\_2} & +1.41 & -2.54 & +1.62 & +1.10 & +1.58 & -3.05 & +1.41 & +1.74 & +1.99 & +2.37 & -2.73\\
\texttt{ext20x20\_3} & +1.54 & -0.78 & +1.94 & +1.03 & +1.75 & -2.45 & +1.54 & +1.87 & +2.21 & +2.64 & -0.86\\
\texttt{ext20x20\_4} & +1.82 & -1.51 & +2.29 & +1.46 & +2.17 & -1.71 & +1.82 & +2.42 & +2.74 & +3.13 & -1.80\\
\midrule
\texttt{cond1e2\_0} & +0.85 & -0.16 & +0.25 & +0.41 & +1.35 & -0.07 & +0.85 & +1.30 & +1.56 & +2.17 & -0.30\\
\texttt{cond1e2\_1} & +1.42 & -2.19 & +1.48 & +1.26 & +1.63 & -2.41 & +1.42 & +1.72 & +2.14 & +2.78 & -2.34\\
\texttt{cond1e2\_2} & +1.53 & -1.36 & +1.01 & +1.29 & +1.57 & -2.19 & +1.53 & +1.90 & +2.69 & +3.08 & -1.61\\
\texttt{cond1e2\_3} & +0.00 & +0.00 & +0.00 & +0.16 & +0.51 & -0.03 & +0.00 & +0.00 & +0.00 & +0.00 & -0.03\\
\texttt{cond1e2\_4} & -0.03 & -0.03 & -0.02 & +0.18 & +0.57 & +0.07 & -0.03 & +0.04 & +0.14 & +0.15 & +0.01\\
\texttt{cond1e4\_0} & +1.63 & +0.49 & +1.64 & +0.86 & +1.74 & +0.48 & +1.63 & +2.18 & +2.82 & +3.15 & +0.09\\
\texttt{cond1e4\_1} & +1.99 & -0.01 & +1.98 & +1.76 & +2.14 & +0.00 & +1.99 & +2.22 & +2.55 & +2.85 & -0.40\\
\texttt{cond1e4\_2} & -0.00 & -0.07 & -0.02 & +0.15 & +0.50 & -0.40 & -0.00 & +0.00 & +0.00 & +0.00 & -0.08\\
\texttt{cond1e4\_3} & +1.40 & -0.89 & +2.43 & +1.24 & +1.57 & -1.46 & +1.40 & +1.68 & +2.03 & +2.86 & -1.66\\
\texttt{cond1e4\_4} & +0.90 & -1.16 & +1.81 & +0.61 & +1.34 & -1.54 & +0.90 & +1.44 & +2.34 & +3.36 & -1.63\\
\texttt{cond1e6\_0} & +0.51 & +0.45 & -0.14 & +0.16 & +0.53 & +0.02 & +0.51 & +0.76 & +1.08 & +1.51 & -0.15\\
\texttt{cond1e6\_1} & +2.78 & +0.40 & +3.19 & +1.85 & +2.86 & -0.45 & +2.78 & +3.08 & +3.34 & +3.65 & +0.19\\
\texttt{cond1e6\_2} & +2.26 & +0.51 & +2.42 & +2.05 & +2.27 & +0.08 & +2.26 & +2.34 & +3.18 & +3.49 & -0.01\\
\texttt{cond1e6\_3} & -0.08 & -0.14 & -0.10 & +0.15 & +0.52 & -0.09 & -0.08 & -0.03 & +0.07 & +0.26 & -0.14\\
\texttt{cond1e6\_4} & +1.47 & +0.44 & +3.03 & +1.12 & +1.76 & -0.28 & +1.47 & +2.61 & +3.18 & +3.52 & -0.23\\
\bottomrule
\end{tabular}
\end{table}

\end{document}